\documentclass[pra,onecolumn,preprintnumbers,superscriptaddress]{revtex4}
\usepackage[a4paper, left=0.8in, right=0.8in, top=0.8in, bottom=0.8in]{geometry}

\pdfoutput=1
\usepackage[utf8]{inputenc}
\usepackage[T1]{fontenc}  
\usepackage{enumerate}
\usepackage{amsmath}
\usepackage{amssymb}
\usepackage{url}
\usepackage{mathtools}
\usepackage{graphics,graphicx,color,float}
\usepackage{etoolbox} 
\usepackage{amsfonts}
\usepackage{amsthm}
\usepackage[usenames,dvipsnames]{xcolor}
\usepackage{tikz}
\usepackage[T1]{fontenc}
\usetikzlibrary{matrix}
\usepackage{xr-hyper}
\usepackage{tikz}
\usetikzlibrary{calc}

\newcommand{\la}{\langle}
\newcommand{\ra}{\rangle}

\newcommand{\calH}{\mathcal{H}}

\newcommand{\R}{\mathbb{R}}
\newcommand{\mcs}{\mathcal{S}}

\newtheorem{fact}{Fact}

\newcommand{\beq}{\begin{equation}}
\newcommand{\enq}{\end{equation}}
\newcommand{\bel}{\begin{lemma}}
\newcommand{\enl}{\end{lemma}}
\newcommand{\bet}{\begin{theorem}}
\newcommand{\ent}{\end{theorem}}



\newcommand{\Tr}{\mathrm{tr}}

\usepackage{booktabs,tabularx}

\newcommand{\suppress}[1]{}

\newcommand{\bra}[1]{\langle #1|}
\newcommand{\ket}[1]{|#1 \rangle}
\newcommand{\braket}[2]{\langle #1|#2\rangle}

\mathchardef\mhyphen="2D

\def\be{\begin{equation}}
\def\ee{\end{equation}}

\newcommand{\gexcl}{\mathcal{G}_{\mathrm{ex}}}

\newcommand{\al}{\alpha}

\makeatletter
\newcommand*{\rom}[1]{\expandafter\@slowromancap\romannumeral #1@}
\makeatother

\makeatletter
\appto{\appendix}{%
 \@ifstar{\def\theequation@prefix{A.}}%
 {}%
}
\makeatother

\mathchardef\mhyphen="2D
\newtheorem{claim}{Claim}
\usepackage[colorlinks,citecolor=red,urlcolor=blue,bookmarks=false,hypertexnames=true]{hyperref} 
\newtheorem{theorem}{Theorem}
\newtheorem{lemma}[theorem]{Lemma}

\newtheorem{definition}[theorem]{Definition}
\newtheorem*{main result}{Main Theorem}

\newtheorem*{theorem*}{Theorem}

\def\01{\{0,1\}}

\newcommand{\im}{\mathop{\rm Im}\nolimits}

\newtheorem{example}{Example}

\theoremstyle{definition}

 \newenvironment{proofof}[1]{\vspace*{5mm} \par 
 \noindent{\it Proof of #1:\hspace{2mm}}}{\qed
}
\renewcommand{\qed}{\hfill$\blacksquare$}
\def\Label#1{\label{#1}\ [\ #1\ ]\ }
\def\Label{\label}

\usepackage{hyperref}
\hypersetup{
    colorlinks=true,
    citecolor = [rgb]{0.3,0.8,0},
    linkcolor=blue,
    filecolor=magenta,      
    urlcolor= magenta,
}

\begin{document}

\title{Graph-Theoretic Framework for Self-Testing in Bell Scenarios}

\begin{abstract}
Quantum self-testing is the task of certifying quantum states and measurements using the output statistics solely, with minimal assumptions about the underlying quantum system. 
It is based on the observation that some extremal points in the set of quantum correlations can only be achieved, up to isometries, with specific states and measurements. Here, we present a new approach for quantum self-testing in Bell non-locality scenarios, motivated by the following observation: the quantum maximum of a given Bell inequality is, in general, difficult to characterize. However, it is strictly contained in an easy-to-characterize set: the \emph{theta body} of a vertex-weighted induced subgraph $(\gexcl,w)$
of the graph in which vertices represent the events and edges join mutually exclusive events. This implies that, for the cases where the quantum maximum and the maximum within the theta body (known as the Lov\'asz theta number) of $(\gexcl,w)$ coincide, self-testing can be demonstrated by just proving self-testability with the theta body of $\gexcl$. This graph-theoretic framework allows us to: (i)~recover the self-testability of several quantum correlations that are known to permit self-testing (like those violating the Clauser-Horne-Shimony-Holt (CHSH) and three-party Mermin Bell inequalities for projective measurements of arbitrary rank, and chained Bell inequalities for rank-one projective measurements), (ii)~prove the self-testability of quantum correlations that were not known using existing self-testing techniques (e.g., those violating the Abner Shimony Bell inequality for rank-one projective measurements). Additionally, the analysis of the chained Bell inequalities, along with prior results in Bell non-locality literature, gives us a closed form expression of the Lov\'asz theta number for a family of well studied graphs known as the Möbius ladders, which might be of independent interest in the community of discrete mathematics. 
\end{abstract}

\author{Kishor Bharti}
\affiliation{Centre for Quantum Technologies, National University of Singapore}

\author{Maharshi Ray}
\affiliation{Department of Information Engineering, Graduate School of Engineering, Mie University, Japan}

\author{Zhen-Peng Xu}
\affiliation{Naturwissenschaftlich-Technische Fakult\"{a}t, Universit\"{a}t Siegen, Walter-Flex-Stra{\ss}e 3, 57068 Siegen, Germany} 

\author{Masahito Hayashi}
\affiliation{Shenzhen Institute for Quantum Science and Engineering, Southern University of Science and Technology, Shenzhen, 518055, China}
\affiliation{Guangdong Provincial Key Laboratory of Quantum Science and Engineering, Southern University of Science and Technology, Shenzhen 518055, China}
\affiliation{Graduate School of Mathematics, Nagoya University, Nagoya, Furo-cho,Chikusa-ku, 464-8602, Japan}
\affiliation{Centre for Quantum Technologies, National University of Singapore}

\author{Leong-Chuan Kwek}
\affiliation{Centre for Quantum Technologies, National University of Singapore} \affiliation{MajuLab, CNRS-UNS-NUS-NTU International Joint Research Unit, Singapore UMI 3654, Singapore}
\affiliation{National Institute of Education, Nanyang Technological University, Singapore 637616, Singapore}

\author{Ad\'{a}n Cabello}
\affiliation{Departamento de F\'{i}sica Aplicada II, Universidad de Sevilla, E-41012 Sevilla, Spain}
\affiliation{Instituto Carlos I de F\'{\i}sica Te\'orica y Computacional, Universidad de Sevilla, E-41012 Sevilla, Spain}

\maketitle

\section{Introduction}

In many information processing tasks, quantum systems render a distinct advantage over classical systems. Motivated by this observation, there has been a rapid development of quantum technologies with potentially new real-world communication and computation applications.
We have also recently witnessed ``quantum supremacy''~\cite{arute2019quantum,zhong2020quantum} and early hints of the quantum internet~\cite{wehner2018quantum}. With the increasing importance of quantum technologies, it becomes pertinent to develop tools for certifying, verifying, and benchmarking quantum devices with minimal assumptions regarding their inner working mechanisms~\cite{eisert2019quantum}. This is a challenging task due to the enormous dimensionality of the Hilbert space associated with the quantum systems. 

One of the prominent approaches to device certification is self-testing
\cite{Yao_self}. The idea of self-testing is to certify underlying
measurement settings and quantum states using solely measurement statistics.
The notion was initially put forward for Bell non-local correlations. The concept has since been extended to
prepare-and-measure scenarios \cite{tavakoli2020self, farkas2019self}, contextuality \cite{BRVWCK19, bharti2019local}, and steering \cite{vsupic2016self, gheorghiu2015robustness,shrotriya2020self}. Self-testing has also been applied to quantum gates and circuits \cite{van2007self,magniez2006self}. A great amount of
work has also been done in making self-testing protocols robust against experimental noise \cite{mckague2012robust,yang2013robust,wu2014robust,miller2013optimal}.
While the Ref. ~\cite{Yao_self} considered the noiseless case for the CHSH self-testing, the authors in ~\cite{mckague2012robust} extended it to the noisy case.
In ~\cite{hayashi2018self}, the authors proposed the mixture of CHSH test and stabilizer test, which has better noise tolerance than the CHSH test.
The authors in ~\cite{mckague2011self} proposed a robust self testing method for graph states. In the aforementioned self-testing protocol by ~\cite{mckague2011self}, the number of required copies increases with order $\mathcal{O}(n^{22})$, where $n$ is the number of qubits of one graph state.
To improve the scaling, the authors in ~\cite{hayashi2018self} proposed another self testing method for the same setting
 with $\mathcal{O}(n^4 \log n)$ copies.
Further, the paper ~\cite{hayashi2019verifiable} proposed a robust self testing protocol for GHZ states. Self-testing with Bell states of higher dimensions has been studied in  ~\cite{kaniewski2019maximal,sarkar2019self}. In ~\cite{kaniewski2016analytic}, tripartite Mermin inequalty was used for robust  self-testing of the three party GHZ state. Robust self-testing protocols based on Chained Bell inequalities have been investigated in Ref. ~\cite{vsupic2016self}.  Comprehensive studies have been carried on for self-testing of single quantum device based on contextuality~\cite{BRVWCK19, bharti2019local} and via computational assumptions ~\cite{metger2020self}. The idea of self-testing has
been used for device-independent randomness generation \cite{coudron2014infinite, colbeck2009quantum, dhara2013maximal, vsupic2016self}, entanglement
detection \cite{bowles2018device, bowles2018self}, delegated quantum computing \cite{reichardt2013classical,mckague2013interactive}, and in several computational complexity proofs, such as the recent breakthrough
result of MIP{*} = RE \cite{ji2020mip}. For a thorough review of self-testing, refer to \cite{vsupic2019self}.

Recently, graph-theoretic techniques have been widely used to study the set of quantum correlations \cite{CSW, slofstra2020tsirelson}. In \cite{CSW}, the authors provide a graph-theoretic characterization of classical and quantum sets in correlation experiments with well-studied objects in graph theory (and combinatorial optimization). In particular, the authors in \cite{CSW} study Bell inequalities and non-contextuality inequalities (a generalization of Bell inequalities). The techniques from  \cite{CSW} have been used to provide robust self-testing schemes in the framework of non-contextuality inequalities for single systems \cite{BRVWCK19}. However, a systematic treatment for Bell scenarios is still lacking. Here, we provide a graph-theoretic approach to study Bell self-testing for multi-partite scenarios by combining techniques from combinatorial optimization and results from \cite{CSW}.

Given a Bell scenario, the set of quantum correlations $B_Q$ is, in general, difficult to characterize. However, $B_Q$ is a strict subset of an easy-to-characterize set, i.e., the \textit{theta body}~\cite{grotschel1988geometric} 
of the graph of exclusivity $\gexcl(V, E)$ of all the events of the scenario. The vertices in $\gexcl(V, E)$ represent the events produced in the scenario \cite{CSW}. The edges in $\gexcl(V, E)$ connect the nodes corresponding to mutually exclusive events.
Using the normalization conditions, every Bell non-locality witness can be written as $S = \sum w_i p_i$, where $w_i > 0$ and $p_i$ are probabilities of events. Therefore, $S$ can be associated to a vertex-weighted graph $(G,w)$ where weights correspond to the $w_i$ and $G$ is an induced subgraph of $\gexcl(V, E)$ \cite{CSW}. The quantum maximum of $S$ must be in the theta body of $G$, which is an even easier to characterize set, as $G$ is a subgraph of $\gexcl(V, E)$. 
Therefore, for the cases where the quantum maximum of a Bell non-locality witness is equal to the maximum of the theta body of $G$, one can prove the self-testing of the Bell inequality by analyzing the theta body of $G$.

We have two sets of assumptions. Our first key assumption  is that the quantum maximum for the Bell witness is equal to the Lov\'asz theta number of the vertex-weighted induced subgraph of $\gexcl(V, E)$ corresponding to the events and their respective weights when the Bell witness is written as a positive linear combination of probabilities of events. The Lov\'asz theta number \cite{lovasz1979shannon} 
is a graph invariant defined in~\eqref{theta:primal} (see section \ref{Back_result}), of the aforementioned induced subgraph.  Our second set of assumptions involve some particular relation among the local projective measurements involved in the scenario. We elaborate on the second set of assumptions in subsection \ref{subsec:results}.  Our results for bipartite and tripartite cases have been stated as Theorems \ref{TH7MT}, \ref{TH8MT}, \ref{TH9MT} and \ref{TH10MT} (see subsection \ref{subsec:results}). Since induced subgraph with weights is still a graph of exclusivity, we will use $(\gexcl,w)$ throughout the paper in place of $(G,w)$.

We apply our techniques to quantum correlations which are known to allow for self testing:  those maximally violating the CHSH~\cite{CHSH}, chained~\cite{pearle70hidden,braunstein90wringling}, and three-party Mermin~\cite{Mermin90} Bell inequaliuties. For CHSH and tripartite Mermin Bell inequalities, we recover self-testing statements for projectors of arbitrary rank. For the family of chained Bell inequalities, we recover self-testing statement for rank-one projective measurements. Our method leads to intriguing insights concerning the dimensionality of the shared quantum state and measurement settings.
We also furnish a self-testing statement for the previously not known case of the Abner Shimony (AS) inequality~\cite{gisin2009bell} for the case of rank-one projective measurements.
In addition, we provide the closed-form expression for the Lov\'asz theta number for M\"{o}bius ladder graphs~\cite{guy1967mobius} 
using the aforementioned connections. The previous closed-form expression was conjectured in~\cite{mateus}. Our result, thus, renders a proof for this conjecture.

The structure of the paper is as follows. We discuss the background literature needed for our work and prove our results in section \ref{Back_result}. The test cases are presented in section \ref{test}. There, we discuss the CHSH, chained, Mermin, and AS Bell inequalities. Finally, in section \ref{discuss}, we discuss the implications of our work and provide some open problems for future study.

\section{Background and Results} \label{Back_result}

\subsection{The graph of exclusivity framework}\label{framework}

A measurement
$M$, together with its outcome $a$, is called a \emph{measurement event} (or event, for
brevity) and denoted $(a|M)$. 
Two events, $e_{i}$ and $e_{j}$ are \emph{mutually exclusive} (or exclusive, for
brevity)
if there exists a measurement $M$ such that $e_{i}$ and $e_{j}$
correspond to different outcomes of $M$. To any set of events $\{e_i\}_{i=1}^N$, we associate a simple undirected graph $\gexcl=([N],E)$, where $[N]$ refers to the set $\{1,2,\ldots, N\}$. This graph, referred to as the {\em graph of exclusivity}, has vertex set $[N]$ and two vertices $i,j$ are adjacent (denoted $i \sim j$) if the corresponding events $e_i$ and $e_j$ are exclusive.

We now consider theories that assign probabilities to events. A {\em behavior} for $\gexcl$ is a mapping $p: [N]\to [0,1]$, such that $p_i+p_j\le 1$, for all $i\sim j$, where we denote $p(i)$ by $p_i$. 
Here, the non-negative scalar $p_i\in [0,1]$ encodes the probability that event $e_i$ occurs. The linear constraint $p_i+p_j\le 1$ enforces that, if $p_i=~1$, then $p_j=~0$.

A behavior $p: [N]\to [0,1]$ is {\em deterministic non-contextual} if all events have pre-determined binary values ($0$ or $1$) that do not depend on the occurrence of other events.  In other words, a deterministic non-contextual behavior $p$ is a mapping $p: [N]\to~\{0,1\}$,such that $p_i+p_j\le 1$, for all $i\sim j$. A {\em deterministic non-contextual} behavior can be considered a vector in $\mathbb{R}^N$.  The convex hull of all deterministic non-contextual behaviors is called the {\em set of non-contextual behaviors}, denoted $\mathcal{P}_{NC}(\gexcl)$. The set  $\mathcal{P}_{NC}(\gexcl)$ is a polytope with its vertices being the deterministic non-contextual behaviors. Behaviors that do not lie in $\mathcal{P}_{NC}(\gexcl)$ are called {\em contextual}. It is worth mentioning that, in combinatorial optimisation, one often encounters the {\em stable set} polytope of a graph $\gexcl$, $\text{\rm STAB}(\gexcl)$ (see Appendix \ref{GTB}). It is quite easy to see that stable sets of $\gexcl$ (a subset of vertices, where no two vertices share an edge between them) and {\em non-contextual} behaviors coincide. 

Lastly, a behavior $p: [N]\to [0,1]$ is called {\em quantum behavior} if there exists a quantum state $\ket{\psi}$ and projectors $\Pi_1,\ldots, \Pi_N$ acting on a Hilbert space $\mathcal{H}$ such that 
\be p_i= \bra{\psi}\Pi_i \ket{\psi}, \forall i\in [N] \text{ and } \Tr(\Pi_i\Pi_j)=0, \text{ for } i\sim j.\ee 
We refer to the ensemble $\ket{\psi}, \{\Pi\}_{i=1}^N$ as a {\em quantum realization} of the behavior $p$. 
The set of all quantum behaviors is a convex set, denoted by $\mathcal{P}_{Q}(\gexcl)$. It turns out that $\mathcal{P}_{Q}(\gexcl)$ is also a well-studied entity in combinatorial optimisation, namely the {\em theta body}, denoted by ${\rm TH}(\gexcl)$ and is formally defined in Appendix \ref{GTB} definition~\ref{Theta1}.

Now, suppose that we are interested in the maximum value of the sum $S = w_1 p_1 + w_2 p_2 + \cdots + w_N p_N$, where $w_i \geq 0$ are weights for $i \in [N]$ and 
\begin{enumerate}
    \item $p \in \mathcal{P}_{NC}(\gexcl)$ is a {\em non-contextual} behavior. In this case, the maximum (henceforth referred to as the classical bound) is given by the independence number of the vertex weighted graph of exclusivity, $\alpha(\gexcl,w)$, that is, the size of the largest clique in the complement graph. Here, $w$ refers to the $N$ dimensional vector of non-negative weights.
    \item $p \in \mathcal{P}_{Q}(\gexcl)$ is a {\em quantum} behavior. In this case, the maximum (henceforth referred to as the quantum bound) is given by the Lov\'asz theta number of the vertex weighted graph of exclusivity, $\vartheta(\gexcl,w)$, defined by the following semidefinite program: 
    
\begin{equation}\label{theta:primal}
\begin{aligned} 
\vartheta(\gexcl, w) = \max & \  \sum_{i=1}^N w_i { X}_{ii} \\
{\rm  s.t.}   &  \  { X}_{ii}={ X}_{0i}, \ \forall i\in [N],\\
  & \ { X}_{ij}=0,\ \forall i\sim j,\\
& \ X_{00}=1,\  X\in \mathbb{S}^{1+N}_+,
\end{aligned}
\end{equation}
\end{enumerate}
where, $\mathbb{S}^{1+N}_+$ denotes positive semidefinite  matrices of size $(N+1) \times (N +1).$ From the definition of the theta body and Lemma~\ref{csdcever} (see Appendix \ref{GTB}), one can note that $p_i = X_{ii}$ for all $i \in [N]$.   

Proofs of the above statements follow quite straightforwardly from the definitions and were first observed in~\cite{CSW}. The Gram-Schmidt decomposition of matrix $X$ corresponding to  \eqref{theta:primal} gives the quantum realization for the underlying behaviour $p$ \cite{BRVWCK19} (see Appendix \ref{GTB} for the definition of Gram-Schmidt decomposition). Note that, for a fixed $X$, its different Gram-Schmidt decompositions are related to one another via isometry.

\begin{definition}(\textbf{Non-contextuality inequality}) For a given graph of exclusivity $\gexcl$, a non-contextuality inequality corresponds to a halfspace that contains the set of non-contextual behaviors, i.e.,
\begin{equation}
\sum_i w_i p_i \leq \alpha(\gexcl, w), \forall p \in \mathcal{P}_{NC}(\gexcl), 
\end{equation}
and $w_i \geq 0$  $ \forall i \in [N]$.
\end{definition}

\subsection{The CHSH experiment in the graph of exclusivity framework}\label{CHSH}

In the CHSH Bell experiment, an arbitrator generates two maximally entangled quantum systems and transmits them to two spatially separated parties: Alice and Bob. Alice has two measurement settings, $x=0$ and $x=1$, and Bob has likewise two measurement settings, $y=0$ and $y=1$. These local measurements are binary observables, each having outcomes, say $0$ and $1$. Each party (Alice and Bob) measures in every round in either the $0$ or the $1$ setting. The selections of settings made by each party must be random and independent of those of the other party. Let $(a,b\vert x,y)$ represent the event where Alice measures in the setting $x$, Bob measures in the setting $y$, and they get $a\in\{0,1\}$ and $b\in\{0,1\}$, respectively. Let the probability of the corresponding event be $p(a,b\vert x,y)$. There are sixteen different events corresponding to all possible combinations of inputs and outputs. They repeat this exercise a considerable enough times, once they are finished,  to determine the probabilities of these events.  

In the CHSH test,eight out of the sixteen events are of relevance, as one is interested in mazimizing the Bell witness given by
\begin{equation}
\label{CHSH}
    S_{\text{CHSH}} = p(0,0\vert0,0) + p(1,1\vert0,1) + p(1,0\vert 1,1) +p(0,0\vert 1,0) +p(1,1\vert0,0) +p(0,0\vert0,1) +p(0,1\vert 1,1) +p(1,1\vert 1,0).  
\end{equation}
Notice that the aforementioned witness  necessitates Alice and Bob to output same answers unless they both are asked $x=y=1$,. In cases they are asked $x=y=1$,  they should opposite answers. The graph of exclusivity corresponding to these eight events is shown in Fig.~\ref{circulantgraph}, and is denoted as $C_{i_8}(1,4)$. The weights on each of the vertices is $1$ and thus the weight vector is an eight dimensional all $1$ vector.  Notice that $\alpha(C_{i_8}(1,4)) = 3$  and, thus, the classical bound of $S_{\text{CHSH}} \leq 3$. Whereas, $\vartheta(C_{i_8}(1,4)) = 2+\sqrt{2} \approx 3.414$, see \cite{CSW}, and therefore the quantum bound of $S_{\text{CHSH}} \leq 2+\sqrt{2}$.

\begin{figure}[h]
\centering
\includegraphics[width=0.4\textwidth]{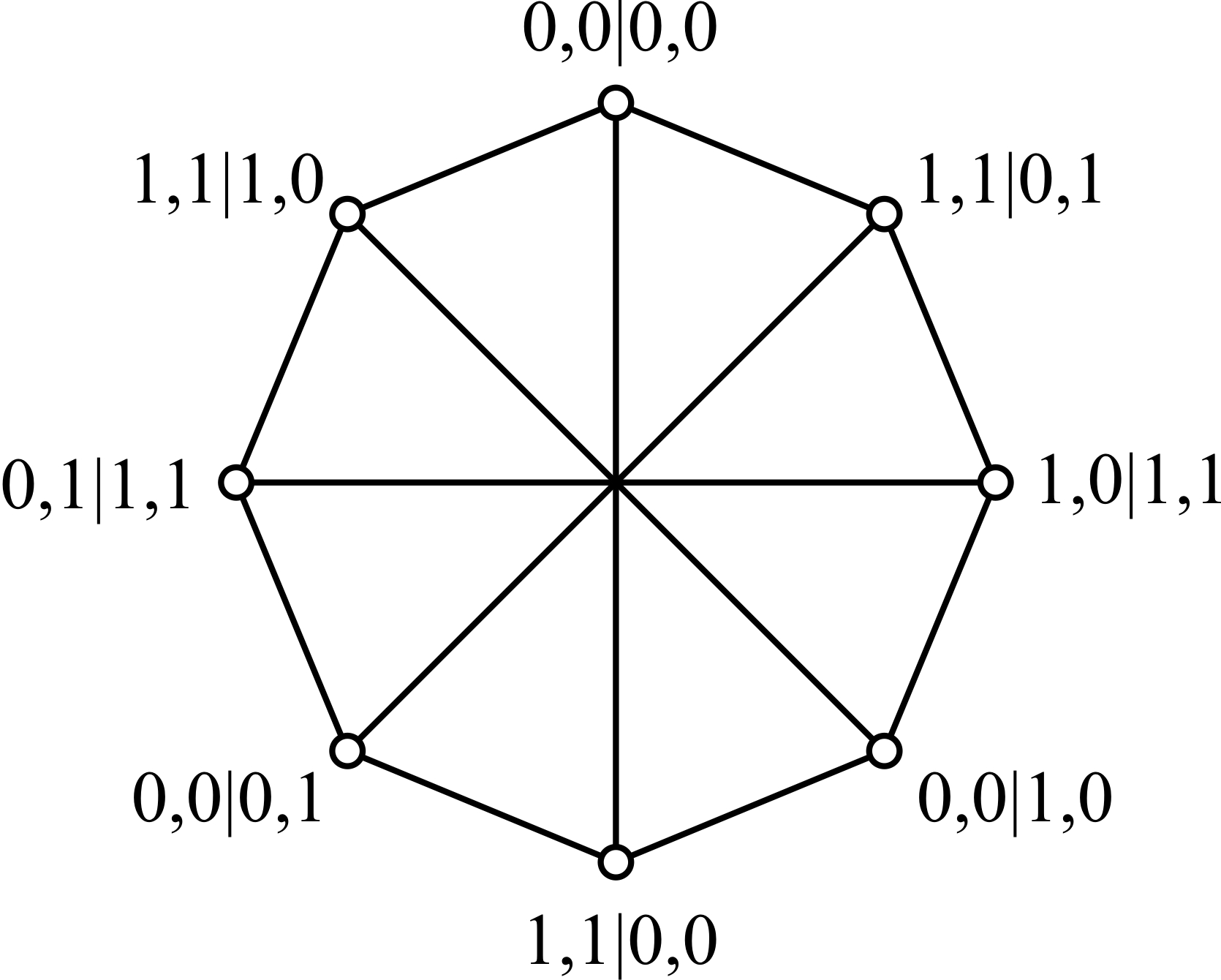}
\caption{Induced subgraph (of the 16-vertex graph of exclusivity of the events in the CHSH scenario) corresponding to the 8 events involved in the expression of the Bell witness given by Eq.~\eqref{CHSH}. This graph is called the 8-vertex circulant graph (1,4) and is denoted $Ci_8[1,4]$ (see definition~\ref{circulantgraph} for a definition of circulant graphs), and is isomorphic to the M\"obius ladder graph of order $2$.}
\centering
\label{circulantgraph}
\end{figure}

\subsection{Self-Testing}

Bell inequalities are special instances of non-contextuality inequalities. Consider an $n$-partite Bell scenario, characterized by a number $n$ of distant observers or parties, their respective measurement settings, and their possible outcomes. Suppose party $j$ possesses $k_j$ different settings with $K_j$ different outcomes for each measurement. In such a scenario, one can compute the probability of a particular string of outcomes given a string of measurements, that is, $p[a_1,a_2,\ldots, a_n | x_1,x_2,\ldots, x_n]$, where $a_j \in [K_j]$ and $x_j \in [k_j]$ for all $j \in [n]$. We use the notation $\vec{a}$ to refer to the $n$-tuple string $a_1,a_2,\ldots,a_n$. Similarly, we use $\vec{x}$ for the measurement settings. A $n$-partite Bell inequality is of the following form:   
\begin{equation}\label{nbell}
\sum_{\vec{a},\vec{x}}s^{\vec{a}}_{\vec{x}} \, p\left[\vec{a} | \vec{x}\right]\le S_{\mathcal{L}}, 
\end{equation}
for some coefficients $s^{\vec{a}}_{\vec{x}}$ and where $S_{\mathcal{L}}$ is the largest possible value allowed in \emph{local hidden variable}(LHV) models~\cite{brunner2014bell}.
The quantum supremum of the Bell expression, i.e., the left hand side of~\eqref{nbell}, denoted by $S_{\mathcal{Q}}$, is the largest possible value of the above expression when $p\left[\vec{a} | \vec{x}\right]$ ranges over the set of quantum behaviors, i.e., 
\begin{equation}
p\left[\vec{a} | \vec{x}\right] = \bra{\psi} \bigotimes_{j = 1}^n M^j_{a_j|x_j}  \ket{\psi},
\end{equation}
for a shared quantum state $\ket{\psi}\in \calH_1 \otimes \calH_2 \ldots,\otimes \calH_n$ and quantum projective measurements $\{ M^j_{a_j|x_j} \}$  acting on $\calH_j$ for all $j \in n$. We refer to the state and the set of measurements that reproduce the quantum behavior, collectively as a \emph{quantum realization}.

\begin{definition}\label{bellstdef} (\textbf{Bell self-testing}) The quantum supremum $S_{\mathcal{Q}}$ of a Bell inequality is a self-test for the realization $(\psi, \{M^j_{a_j|x_j}\}_j)$, if for any other realization $({\psi}', \{M'^j_{a_j|x_j} \})$ that also attains $S_{\mathcal{Q}}$, there exists a local unitary $V=V_1\otimes V_2,\ldots,\otimes V_n$ and an ancilla state $\ket{junk}$ such~that 
\be\label{nbell:ST}
\begin{aligned}
V\ket{\psi'}&=\ket{junk}\otimes \ket{\psi} ,\\
V(\bigotimes_{j = 1}^n M'^j_{a_j|x_j}) \ket{\psi'} & =\ket{junk}\otimes (\bigotimes_{j = 1}^n M^j_{a_j|x_j})\ket{\psi}.
\end{aligned}
\ee
\end{definition} 

\subsection{Relevant background from Semidefinite programs}\label{SDP_back}
\begin{definition}(\textbf{Semidefinite programs}) A pair of primal and dual SDPs is given by an optimisation problem of the following form:
\begin{align} 
& \underset{{ X}}{\text{sup}} \left\{ \la C,{ X}\ra : X \in \mcs^n_+,\ \la A_i,{ X} \ra=b_i \ (i\in [m]) \right\},\\
&\underset{y,Z}{\text{inf}}\left\{ \sum_{i=1}^m b_iy_i\ :\ \sum_{i=1}^my_iA_i-C=Z\in \mcs^n_+\right\},
\end{align}
where $C, A_i$ (for all $i \in [m]$) are Hermitian $n \times n$ matrices and $b  \in \mathbb{C}^m$.
\end{definition}
We have introduced the primal formulation of the Lov\'asz theta SDP in~\eqref{theta:primal}.

The dual formulation  for \eqref{theta:primal} is given by
\be \label{theta:dual}
\min t :   tE_{00}+\sum_{i=1}^n (\lambda_i-1)E_{ii}-\sum_{i=1}^n\lambda_i E_{0i}+\sum_{i\sim j} \mu_{ij}E_{ij} \equiv Z \succeq 0,
\ee
where $E_{ij} = \frac{e_ie_j^T + e_je_i^T }{2}$. We make crucial use of the following Theorem due to Alizadeh {\em et al.}~\cite[Theorem 4]{alizadeh} to show that the optimiser of~\eqref{theta:primal}  is unique.  

\begin{theorem}~\cite{alizadeh} \label{alizadeh}
Let $Z^*$ be a dual optimal and nondegenerate solution of a semidefinite program. Then, there exists a unique primal optimal solution for that SDP. 
\end{theorem}

The notion of dual nondegeneracy is given by the following definition.
\begin{definition}(\textbf{Dual nondegeneracy}) Let $Z^*$ be an optimal dual solution and let $M$ be any symmetric matrix. If the homogeneous linear system 

\begin{align}
&MZ^* = 0, \label{mzcon} \\
&\Tr(MA_i) = 0 \ (\forall i \in [m]), \label{macon}    
\end{align}
only admits the trivial solution $M=0$, then $Z^*$ is said to be dual nondegenerate. 
\end{definition}

A key ingredient for proving the results in this paper is the following lemma: 

\begin{lemma}(\cite{BRVWCK19}) \label{stprl}
 Let $X^{*}$ be the unique optimal solution for the primal and let $\left\{\left|u_{i}\right\rangle\left\langle u_{i}\right|\right\}_{i=0}^{n}$ be a quantum realization achieving the maximum quantum value of $\sum_{i=1}^{n} w_{i} p_{i}: p \in \mathcal{P}_{q}\left(\mathcal{G}_{\mathrm{ex}}\right)$. Then, the non-contextuality inequality $\sum_{i=1}^{n} w_{i} p_{i} \leq B_{n c}\left(\mathcal{G}_{\mathrm{ex}}, w\right)$, for all $p \in$
$\mathcal{P}_{n c}\left(\mathcal{G}_{\mathrm{ex}}\right)$ is a self-test for the realization $\left\{\left|u_{i}\right\rangle\left\langle u_{i}\right|\right\}_{i=0}^{n}$.
\end{lemma}

\subsection{Results} \label{subsec:results}

We are given a Bell inequality of the form~\ref{nbell} and we consider the set of events $p\left[\vec{a} | \vec{x}\right]$ such that $s^{\vec{a}}_{\vec{x}} \neq 0$. We shall index this set by $i$ and denote the corresponding event as $e_i$. Suppose we are given a $n$-partite Bell inequality with Bell witness $\mathcal{B}=\sum_i w_i p_i$, with $w_i > 0$ and $p_i=p(e_i)$, and a quantum realization $(\psi, \{M^j_{a_j|x_j}\}_j)$ (let us call this the {\em reference} system) that achieves the quantum supremum, $S_{\mathcal{Q}}$ of $\mathcal{B}$. Let $(\gexcl,w)$ be the weighted graph capturing the weights $\{w_i\}$ and mutual exclusivity relationships among the events $\{e_i\}$ in $\mathcal{B}$. 

We have two sets of assumptions in this manuscript. The first set of assumptions is following: 
\begin{enumerate}[(i)]
    \item $S_{\mathcal{Q}} = \vartheta(\gexcl,w)$. \label{ass1}
    \item The Lov\'asz theta SDP in~\eqref{theta:primal} corresponding to $(\gexcl,w)$ has a unique maximizer. This assumption is a consequence of assumption \ref{ass1} for the scenarios of interest in this paper. \label{ass2}
\end{enumerate}

We consider two types of sets of indexes ${\cal I}$ and ${\cal I}_0={\cal I} \cup \{0\}$.
We consider 
the matrix $X_{ij}:=\langle \psi | \Pi_j \Pi_i |\psi\rangle$,
where $\Pi_i$ is a projection and 
$\Pi_0$ is the identity operator.
We set $n:=|{\cal I}|$.
The assumption \ref{ass2} means that the following SDP has unique solution.
\begin{equation}\label{thetaSDP1}
\begin{aligned} 
\vartheta(\gexcl,w) = \max & \  \sum_{i \in {\cal I}} w_i { X}_{ii} \\
{\rm  s.t.}   &  \  { X}_{ii}={ X}_{0i}, \ \forall i\in [n],\\
  & \ { X}_{ij}=0,\ \forall i\sim j,\\
& \ X_{00}=1,\  X\in \mathbb{S}^{1+n}_+.
\end{aligned}
\end{equation}

The second set of assumptions depend on the scenarios of interest and have been mentioned in the following subsubsections. Our results for bipartite and tripartite cases have been summarized as Theorems \ref{TH7MT}, \ref{TH8MT}, \ref{TH9MT} and \ref{TH10MT}.
\subsubsection{Bipartite case}
Suppose the unique optimal maximizer 
$X^*=(X_{ij})$ is given by $\eta_i \eta_j\langle v_j, v_i\rangle$ with the following;
For $i=(i_A,i_B)\in {\cal I}$, 
\begin{align} v_{i} = a_{i_A} \otimes b_{i_B}, 
\label{NA1MT}
\end{align}
where $a_{i_A} \in {\cal H}_A=\mathbb{C}^{d_A}$, 
$b_{i_B} \in {\cal H}_B=\mathbb{C}^{d_B}$.
Also, for simplicity, $a_{i_A} $ and $ b_{i_B}$ are assumed to be normalized
and $\eta_i>0$. Now, we consider a state $|\psi'\rangle $ on ${\cal H}_A'\otimes {\cal H}_B'$,
and 
projections $\Pi_{i_A}^A$ and $\Pi_{i_B}^B$ on ${\cal H}_A'$ and ${\cal H}_B'$.
Here, 
when $i_A=i_A'$ ($i_B=i_B'$) for $i \neq i'$,
$\Pi_{i_A}^A=\Pi_{i_A'}^A$ ($\Pi_{i_B}^B=\Pi_{i_B'}^B$).
Then, we define the projection 
$\Pi_i:=\Pi_{i_A}^A\otimes \Pi_{i_B}^B$,

In the following, 
we discuss how 
the state $ |\psi'\rangle$ is locally converted to 
$ |\psi\rangle$ when
the vectors $\Pi_i |\psi'\rangle$ realize the optimal solution in the SDP \eqref{thetaSDP1}.
We define $|v_i'\rangle:= \eta_i^{-1} \Pi_{i}|\psi'\rangle$.

First, we consider the case that the ranks of the projections $\Pi_{i_A}^A$ and $\Pi_{i_B}^B$ 
are one.
We introduce the following conditions.
\begin{description}
\item[A1]
The set $\{v_i\}_{i\in {\cal I}_0}$ of vectors 
span the vector space ${\cal H}_A\otimes {\cal H}_B$.

\item[A2]
There exist
a subset ${\cal I}_B$ of indecies of the space ${\cal H}_B$ with
$|{\cal I}_B|=d_B=\dim {\cal H}_B$
and 
$d_B$ sets $\{ {\cal I}_{A,i_B} \}_{i_B \in {\cal I}_B}$ of indecies of the space 
${\cal H}_A$ 
$|{\cal I}_{A,i_B}|=d_A=\dim {\cal H}_A$
to satisfy the following conditions B1-B4.
\end{description}

\begin{description}
\item[B1]
$\cup_{i_B \in {\cal I}_B} {\cal I}_{A,i_B}\times \{i_B\}
\subset {\cal I}$.
\item[B2]
$\{ b_{i_B} \}_{i_B \in {\cal I}_B}$ spans the space ${\cal H}_B$.
\item[B3]
$\{ a_{i_A} \}_{i_A \in {\cal I}_{A, i_B}}$ spans the space ${\cal H}_A$
for any $i_B \in {\cal I}_B$.
\item[B4]
We define the graph on ${\cal I}_B$ in the following way.
The node
$i_B\in {\cal I}_B$ is connected to $i_B' \in {\cal I}_B$ 
when the following two conditions holds.
\begin{description}
\item[B4-1]
The relation $\langle b_{i_B},b_{i_B'}\rangle \neq 0$ holds.
\item[B4-2]
The relation ${\cal I}_{A, i_B}\cap  {\cal I}_{A, i_B'} \neq \emptyset$ holds.

\end{description}
\end{description}

In the two qubit case, if 
the set $\{v_i\}_{i\in {\cal I}}$ of vectors 
contains the following $4$ vectors, then  
the conditions A1 and A2 hold;
\begin{align}
a_0 \otimes b_0,~
a_1 \otimes b_0,~
a_0 \otimes b_1,~
a_2 \otimes b_1,\label{XO1MT}
\end{align}
where $a_0\neq a_1,a_2$,
$ \langle b_0,b_1\rangle \neq 0$.

\begin{theorem}\Label{TH7MT}
Assume that 
the optimal maximizer given in \eqref{NA1MT} satisfies conditions A1 and A2
and
the vectors $(\Pi_i |\psi'\rangle)_{i \in {\cal I}}$ realize the optimal solution in the SDP \eqref{thetaSDP1}.
In addition, the ranks of the projections $\Pi_{i_A}^A$ and $\Pi_{i_B}^B$ 
are assumed to be one.
Then, there exist isometries 
$V_A: {\cal H}_A \to {\cal H}_A'$
and $V_B: {\cal H}_B\to {\cal H}_B'$ such that
\begin{align}
V_A \otimes V_B|\psi \rangle&= |\psi'\rangle , \\
V_A \otimes V_B|v_i \rangle&= |v_i'\rangle ,
\end{align}
for $i \in {\cal I}$.
\hfill $\square$\end{theorem}

\begin{proof} 
The proof has been deferred to Appendix \ref{ST_Proof}.
\end{proof}

Now we consider the general case.
In addition to A1 and A2, we assume the following condition.
\begin{description}
\item[A3] Ideal systems ${\cal H}_A$ and ${\cal H}_B$ are two-dimensional.
\item[A4] Each system has only two measurements.
That is, the set $\bar{\cal I}_A$ ($\bar{\cal I}_B$) of all indexes of the space ${\cal H}_A$ (${\cal H}_B$) is composed of
$4$ elements.
For any element $i_A \in \bar{\cal I}_A$ ($i_B \in \bar{\cal I}_B$), there exists 
an element $i_A' \in \bar{\cal I}_A$ ($i_B' \in \bar{\cal I}_B$) such that 
$\langle a_{i_A}| a_{i_A'}\rangle=0$
($\langle b_{i_B}| b_{i_B'}\rangle=0$).
\end{description}

When A3 and A4 hold,
$\bar{\cal I}_A$ ($\bar{\cal I}_B$) is written as
${\cal B}_{A,0}\cup {\cal B}_{A,1}$ (${\cal B}_{B,0}\cup {\cal B}_{B,1}$), where
${\cal B}_{A,j}=\{ (0,j),(1,j)\}$ (${\cal B}_{B,j}=\{ (0,j),(1,j)\}$) 
and $\langle a_{(0,j)}| a_{(1,j)}\rangle=0$ ($\langle b_{(0,j)}| b_{(1,j)}\rangle=0$) 
for $j=0,1$.

We also consider the following condition for $\Pi_i= \Pi_{i_A}^A\otimes \Pi_{i_B}^B$.
\begin{description}
\item[C1]
When $i_A,i_A' \in \bar{\cal I}_A$ ($i_B,i_B' \in \bar{\cal I}_B$) satisfy
$\langle a_{i_A}| a_{i_A'}\rangle=0$ ($\langle b_{i_B}| b_{i_B'}\rangle=0$), 
we have $\Pi_{i_A}^A+\Pi_{i_A'}^A=I$ ($\Pi_{i_B}^B+\Pi_{i_B'}^B=I$). 
\end{description}

Let ${\cal H}_{i_A}^A$ and ${\cal H}_{i_B}^B$ be the image of the projections
$\Pi_{i_A}^A$ and $\Pi_{i_B}^B$.

\begin{theorem}\Label{TH8MT}
Assume that 
the optimal maximizer given in \eqref{NA1MT} satisfies conditions A1, A2, A3, and A4,
the vectors $(\Pi_i |\psi'\rangle)_{i \in {\cal I}}$ realize the optimal solution in the SDP \eqref{thetaSDP1},
and condition C1 holds.
Then, there exist isometries 
$V_A$ from $ {\cal H}_A\otimes  {\cal K}_A$ to $ {\cal H}_A'$
and 
$V_B$ from $ {\cal H}_B\otimes  {\cal K}_B$ to $ {\cal H}_B'$
such that
\begin{align}
V_A \otimes V_B |\psi\rangle \otimes |junk\rangle &=
|\psi'\rangle , \Label{ML1}\\
V_A \otimes V_B |v_i\rangle \otimes |junk\rangle &=|v_i'\rangle , \Label{ML2}
\end{align}
for $i \in {\cal I}$,
where
$|junk\rangle$ is a state on ${\cal K}_A\otimes  {\cal K}_B$.
\hfill $\square$\end{theorem}

\begin{proof} 
The proof has been deferred to Appendix \ref{ST_Proof}.
\end{proof}

\subsubsection{Tripartite case}
We assume that the unique optimal maximizer 
$X^*=(X_{ij})$ is given by $\eta_i \eta_j\langle v_j, v_i\rangle$ with the following;
For $i=(i_A,i_B,i_C)\in {\cal I}$, 
\begin{align}
 v_{i} = a_{i_A} \otimes b_{i_B}\otimes c_{i_C}, 
\label{MA8MT}
\end{align}
where $a_{i_A} \in {\cal H}_A=\mathbb{C}^{d_A}$, 
$b_{i_B} \in {\cal H}_B=\mathbb{C}^{d_B}$,
$c_{i_C} \in {\cal H}_C=\mathbb{C}^{d_C}$.
Also, for simplicity, $a_{i_A} $, $ b_{i_B}$, and $ c_{i_C}$ 
are assumed to be normalized and $\eta_i>0$. Now, we consider a state $|\psi'\rangle $ on 
${\cal H}_A'\otimes {\cal H}_B'\otimes {\cal H}_C'$,
and 
projections $\Pi_{i_A}^A$, $\Pi_{i_B}^B$, $\Pi_{i_C}^C$ 
on ${\cal H}_A'$, ${\cal H}_B'$, and ${\cal H}_C'$.
Then, we define the projection 
$\Pi_i:=\Pi_{i_A}^A\otimes \Pi_{i_B}^B\otimes \Pi_{i_C}^B$.

In the following, 
we discuss how 
the state $ |\psi'\rangle$ is locally converted to 
$ |\psi\rangle$ when
the vectors $\Pi_i |\psi'\rangle$ realize the optimal solution in the SDP \eqref{thetaSDP1}.
We define $|v_i'\rangle:= \eta_i^{-1} \Pi_{i}|\psi'\rangle$.

We consider the case that the ranks of the projections 
$\Pi_{i_A}^A$, $\Pi_{i_B}^B$ and $\Pi_{i_C}^C$ 
are one.
We introduce the following conditions.

\begin{definition}\Label{D1MT}
Three distinct elements $i,j,k\in {\cal I}$ are called {\it linked}
when the following two conditions holds.
\begin{description}
\item[C1]
The relations 
$\langle v_i,v_{k}\rangle \neq 0$,
$\langle v_i,v_{j}\rangle \neq 0$, and
$\langle v_j,v_{k}\rangle \neq 0$
hold.
\item[C2]
$v_i,v_{j}$ shares a $t_{i,j}-$th common element for $t_{i,j} \in \{A,B,C\}$.
Other components of $v_i,v_{j}$ are different.
That is, when $t_{i,j}=A$, $i_A=j_A$,$i_B\neq j_B$,and $i_C\neq j_C$.
$v_i$ and $v_{k}$ share a $t_{i,k}-$th common element for 
$t_{i,k} \in \{A,B,C\}\setminus \{t_{i,j}\}$.
$v_j,v_{k}$ shares a $t_{j,k}-$th common element for 
$t_{j,k} \in \{A,B,C\}\setminus \{t_{i,j},t_{i,k}\}$.
In this case, there exist elements 
$x_A,x_A',x_B,x_B',x_C,x_C'$ such that
$i,j,k \in \{x_A,x_A'\}\times \{x_B,x_B'\} \times \{x_C,x_C'\}$.
\end{description}
In addition, 
two distinct elements $x_A,x_A'$ for index of a vectors of $\mathbb{C}^{d_A}$
are called {\it connected}
when there exist three linked elements $i,j,k\in {\cal I}$ 
such that the first components of $i,j,k\in {\cal I}$ are $x_A,x_A'$.
\hfill $\square$\end{definition}

For $i_B,i_C$, we use notation
\begin{align}
\psi_{(i_B,i_C)}:= b_{i_B} \otimes c_{i_C}.
\end{align}
Then, we introduce the following conditions for the optimal maximizer given in \eqref{MA8MT}.

\begin{description}
\item[A5]
The vectors $\{v_i\}_{i\in {\cal I}_0}$ 
span the vector space ${\cal H}_A\otimes {\cal H}_B\otimes {\cal H}_C$.

\item[A6]
There exist
a subset ${\cal I}_A$ of indecies of the space 
${\cal H}_A$ with $|{\cal I}_A|=d_A$
and 
$d_A$ sets ${\cal I}_{BC,i_A}$ for $i_A \in {\cal I}_A$ of indecies of the space 
${\cal H}_B\otimes {\cal H}_C$ 
to satisfy the following conditions.
The set $\{a_{i_A}\}_{i_A\in {\cal I}_A}$
spans the space ${\cal H}_A$.
The set $\{\psi_{i_{BC}}\}_{i_{BC}\in {\cal I}_{BC,i_A}}$
spans the space ${\cal H}_B\otimes {\cal H}_C$ and 
${\cal I}_0= \cup_{i_A \in {\cal I}_A} (\{i_A\} \times {\cal I}_{BC,i_A})$.
We consider the graph $G_A$ with the set ${\cal I}_A$ of vertecies 
such that 
the edges are given as the the pair of 
all connected elements in ${\cal I}_A$ in the sense of the end of Definition \ref{D1MT}.
The graph $G_A$ is not divided into two disconnected parts.

\item[A7]
The vectors
$\{b_{i_B}\otimes c_{i_C}\}_{
(i_B,i_C)\in \cup_{i_A \in {\cal I}_A}{\cal I}_{BC,i_A}}$ 
satisfy condition A2 by substituting $c_{i_C}$ into $a_{i_A}$.
That is, there exist a subset ${\cal I}_B$ of the second indecies and
subsets ${\cal I}_{C,i_B}$ of the third indecies such that
they satisfy conditions B1, B2, B3, and B4.
We denote the graph defined in this condition by $G_B$

\end{description}

\begin{theorem}\Label{TH9MT}
Assume that the optimal maximizer given in \eqref{MA8MT} satisfies
conditions A5, A6, and A7, and
the vectors $(\Pi_i |\psi'\rangle)_{i \in {\cal I}}$ realize the optimal solution in the SDP \eqref{thetaSDP1}.
In addition, the ranks of the projections 
$\Pi_{i_A}^A$, $\Pi_{i_B}^B$, and $\Pi_{i_C}^C$ 
are assumed to be one.

Then, there exist isometries 
$V_A: {\cal H}_A \to {\cal H}_A'$,
$V_B: {\cal H}_B\to {\cal H}_B'$, 
and $V_C: {\cal H}_C\to {\cal H}_C'$ 
such that
\begin{align}
V_A \otimes V_B\otimes V_C|\psi \rangle=&
|\psi'\rangle , \\
V_A \otimes V_B\otimes V_C|v_i \rangle=&
|v_i'\rangle ,
\end{align}
for $i \in {\cal I}$.
\end{theorem}
\begin{proof}
The proof has been deferred to Appendix \ref{ST_Proof}.
 \end{proof}

Now we consider the general case.
We define 
$|v_i'\rangle:= \eta_i^{-1} \Pi_{i_A}^A\otimes \Pi_{i_B}^B
\otimes \Pi_{i_C}^C|\psi'\rangle$.

Let $\bar{\cal I}_A,\bar{\cal I}_B,\bar{\cal I}_C$ be the sets 
of indecies of the spaces 
${\cal H}_A,{\cal H}_B, {\cal H}_C$.

We introduce other conditions for 
the optimal maximizer given in \eqref{MA8MT} as a generalization of A3 and A4.

\begin{description}
\item[A8] Ideal systems ${\cal H}_A$, ${\cal H}_B$, and ${\cal H}_C$ are two-dimensional.
\item[A9] Each system has only two measurements.
That is, the sets $\bar{\cal I}_A$, $\bar{\cal I}_B$, and $\bar{\cal I}_C$
is composed of
4 elements.
For any element $i_A \in \bar{\cal I}_A$ ($i_B \in \bar{\cal I}_B$, $i_C \in \bar{\cal I}_C$), there exists 
an element $i_A' \in \bar{\cal I}_A$ ($i_B' \in \bar{\cal I}_B$, $i_C' \in \bar{\cal I}_C$) such that 
$\langle a_{i_A}| a_{i_A'}\rangle=0$
($\langle b_{i_B}| b_{i_B'}\rangle=0$, $\langle c_{i_C}| c_{i_C'}\rangle=0$).
\end{description}

When A3 and A4 hold,
$\bar{\cal I}_A$ ($\bar{\cal I}_B$, $\bar{\cal I}_C$) is written as
${\cal B}_{A,0}\cup {\cal B}_{A,1}$ (${\cal B}_{B,0}\cup {\cal B}_{B,1}$, 
${\cal B}_{C,0}\cup {\cal B}_{C,1}$), where
${\cal B}_{A,j}=\{ (0,j),(1,j)\}$ (${\cal B}_{B,j}=\{ (0,j),(1,j)\}$, ${\cal B}_{C,j}=\{ (0,j),(1,j)\}$) 
and $\langle a_{(0,j)}| a_{(1,j)}\rangle=0$ 
($\langle b_{(0,j)}| b_{(1,j)}\rangle=0$, $\langle c_{(0,j)}| c_{(1,j)}\rangle=0$) 
for $j=0,1$.

We also consider the following condition for $\Pi_i= \Pi_{i_A}^A\otimes \Pi_{i_B}^B\otimes \Pi_{i_C}^C$.
\begin{description}
\item[C1]
When $i_A,i_A' \in \bar{\cal I}_A$ ($i_B,i_B' \in \bar{\cal I}_B$, 
$i_C,i_C' \in \bar{\cal I}_C$) satisfy
$\langle a_{i_A}| a_{i_A'}\rangle=0$ 
($\langle b_{i_B}| b_{i_B'}\rangle=0$, $\langle c_{i_C}| c_{i_C'}\rangle=0$), 
we have $\Pi_{i_A}^A+\Pi_{i_A'}^A=I$ ($\Pi_{i_B}^B+\Pi_{i_B'}^B=I$, 
$\Pi_{i_C}^C+\Pi_{i_C'}^C=I$). 
\end{description}

Let ${\cal H}_{i_A}^A$, ${\cal H}_{i_B}^B$, and ${\cal H}_{i_C}^C$ be the image of the projections
$\Pi_{i_A}^A$, $\Pi_{i_B}^B$, and $\Pi_{i_C}^C$.

\begin{theorem}\Label{TH10MT}
Assume that the optimal maximizer given in \eqref{MA8MT} satisfies
conditions A5, A6, A5, A7, A8, and A9, and 
the vectors $(\Pi_i |\psi'\rangle)_{i \in {\cal I}}$ realize the optimal solution in the SDP \eqref{thetaSDP1}.
Then, there exist isometries 
$V_A$ from $ {\cal H}_A\otimes  {\cal K}_A$ to $ {\cal H}_A'$,
$V_B$ from $ {\cal H}_B\otimes  {\cal K}_B$ to $ {\cal H}_B'$,
and 
$V_C$ from $ {\cal H}_C\otimes  {\cal K}_C$ to $ {\cal H}_C'$
such that
\begin{align}
V_A \otimes V_B\otimes V_C |\psi\rangle \otimes |junk\rangle &=
|\psi'\rangle,\Label{ML5}\\
V_A \otimes V_B \otimes V_C|v_i\rangle \otimes |junk\rangle &=|v_i'\rangle , \Label{ML6}
\end{align}
for $i \in {\cal I}$,
where
$|junk\rangle$ is a state on ${\cal K}_A\otimes  {\cal K}_B\otimes  {\cal K}_C$.

\hfill $\square$\end{theorem}

\begin{proof}
The proof has been deferred to Appendix \ref{ST_Proof}.
\end{proof}

\section{Test Cases} \label{test}

In the following subsections we apply our techniques to the CHSH, chained, Mermin, and AS Bell inequalities.

\subsection{CHSH Self-Testing}

Self-testing is known to hold for the maximum quantum violation of the CHSH inequality \cite{Yao_self}. Here, we study
the CHSH inequality in the graph of exclusivity framework~\cite{CSW}.  

\medskip

Recall that the graph of exclusivity corresponding to the Bell witness given by Eq.~\eqref{CHSH}
is given by the $C_{i_8}(1,4)$ graph (see figure~\ref{circulantgraph}). We claim that the optimal solution to dual~\eqref{theta:dual} for  $C_{i_8}(1,4)$ is given by 
\be\label{dualCHSH}
Z_{\text{CHSH}}=\left(\begin{array}{c|cccccccc} 
2+\sqrt{2} & -1 & -1 & -1 & -1 & - 1&-1& - 1&-1\\
\hline
-1 & 1 & h& 0 & 0 & k & 0& 0 & h\\
-1 & h & 1 & h& 0 & 0& k & 0 & 0 \\
-1 & 0 & h & 1 & h& 0& 0 & k& 0\\
-1 & 0 & 0 & h& 1 & h& 0 & 0 & k\\
-1& k & 0 & 0& h & 1& h & 0 & 0\\
-1 & 0 & k & 0 & 0& h& 1 & h & 0\\
-1 & 0 & 0 & k& 0& 0& h & 1 & h\\
-1& h & 0 & 0& k & 0& 0 & h & 1
\end{array}\right), 
\ee
where $k = 3 -2\sqrt{2}$ and $h =2- \sqrt{2}$. Lov\'asz theta SDP has zero duality gap, that is, the primal optimal solution and optimal dual solution yield the same program value. It can be easily verified that \eqref{dualCHSH} is a feasible solution to~\eqref{theta:dual} for the graph $C_{i_8}(1,4)$. The dual solution \eqref{dualCHSH} achieves $2 + \sqrt{2} $ and is thus dual optimal. In order to show the uniqueness of the primal optimal, we show that $Z_{\text{CHSH}}$ is nondegenerate. That requires us to show that $M=0$ is the only symmetric $9\times 9$ matrix satisfying equations~\eqref{mzcon} and~\eqref{macon} corresponding to the Lov\'asz theta SDP. That is, the linear system    
\be 
M_{00}=0, \ M_{0i}=M_{ii}, \ M_{ij}=0\ (\forall \, i \sim j), \ MZ^*=0
\ee
has a unique solution $M=0$. Barring the $MZ^* =0$ constraint, the rest of the constraints already guarantee that several entries of $M$ must be zeros. Thus the $M$ matrix has the following form:
\be 
M=\left(\begin{array}{c|cccccccc} 
0& m_1 & m_2& m_3& m_4& m_5&m_6& m_7&m_8\\
\hline
m_1 & m_1 & 0& m _9& m_{15}& 0 & m_{20}& m_{23} & 0\\
m_2 & 0 & m_2 & 0&m_{10}  & m_{16}& 0 & m_{21} & m_{24} \\
m_3& m_{9} & 0 & m_3 & 0& m_{11}& m_{17}& 0& m_{22}\\
m_4 & m_{15} & m_{10} & 0& m_4 & 0& m_{12}& m_{18} & 0\\
m_5& 0 & m_{16} & m_{11}& 0 & m_5& 0 & x _{13}& m_{19}\\
m_6 & m_{20} & 0 & m_{17} & m_{12}& 0& m_6 & 0 & m_{14}\\
m_7 & m_{23} & m_{21} & 0& m_{18}& m_{13}& 0 & m_7 & 0\\
m_8& 0 & m_{24} & m_{22}& 0 & m_{19}& m_{14} & 0 & m_8
\end{array}\right).
\ee
It can be easily checked that the only solution to the system of linear equations $M{ }\cdot Z_{\text{CHSH}}=0$ is $M=0$. 

\medskip

The optimal solution for primal~\eqref{theta:primal} is given by
\be\label{primalCHSH}
P_{\text{CHSH}}=\left(\begin{array}{c|cccccccc} 
1 & \chi & \chi & \chi & \chi & \chi & \chi& \chi&\chi\\
\hline
\chi & \chi & 0& \frac{1}{2}\chi & \xi  & 0 & \xi&  \frac{1}{2}\chi & 0\\
\chi & 0& \chi & 0& \frac{1}{2}\chi & \xi & 0 & \xi &  \frac{1}{2}\chi\\
\chi & \frac{1}{2}\chi & 0 & \chi & 0& \frac{1}{2}\chi & \xi  & 0& \xi\\
\chi & \xi & \frac{1}{2}\chi & 0& \chi & 0& \frac{1}{2}\chi & \xi  & 0\\
\chi& 0 & \xi& \frac{1}{2}\chi & 0 & \chi& 0 & \frac{1}{2}\chi & \xi \\
\chi & \xi & 0 & \xi & \frac{1}{2}\chi & 0& \chi & 0 & \frac{1}{2}\chi\\
\chi & \frac{1}{2}\chi & \xi & 0& \xi& \frac{1}{2}\chi & 0 & \chi & 0\\
\chi& 0 & \frac{1}{2}\chi & \xi& 0 & \xi& \frac{1}{2}\chi & 0 & \chi
\end{array}\right), 
\ee
where $\chi= \frac{2+\sqrt{2}}{8}$ and $\xi = \frac{1+\sqrt{2}}{8}$. The configurations corresponding to the primal optimal matrix $P_{\text{CHSH}}$ correspond to different Gram decomposition of $P_{\text{CHSH}}$ and are related to each other via global isometry. A quantum realization is achieved with the two-qubit maximally entangled state $\ket{\psi} = (\frac{1}{\sqrt{2}},0,0,\frac{1}{\sqrt{2}})^T$
and the vectors corresponding to the $8$ projective measurements given by
\begin{equation}
\begin{aligned} 
\ket{v_1} &= \ket{A_{1,0}} \otimes \ket{B_{1,1}}, \\
\ket{v_2} &= \ket{A_{0,0}} \otimes \ket{B_{1,0}}, \\
\ket{v_3} &= \ket{A_{0,1}} \otimes \ket{B_{0,1}}, \\
\ket{v_4} &= \ket{A_{1,0}} \otimes \ket{B_{0,0}}, \\
\ket{v_5} &= \ket{A_{1,1}} \otimes \ket{B_{1,0}}, \\
\ket{v_6} &= \ket{A_{0,1}} \otimes \ket{B_{1,1}}, \\
\ket{v_7} &= \ket{A_{0,0}} \otimes \ket{B_{0,0}}, \\
\ket{v_8} &= \ket{A_{1,1}} \otimes \ket{B_{0,1}},
\end{aligned}    
\end{equation}
where the kets corresponding to the local measurements are given by
\begin{equation}
\begin{aligned}
\ket{A_{0,0}} &= \left(1,0\right)^T, \\
\ket{A_{0,1}} &= \left(0,-1\right)^T, \\
\ket{A_{1,0}} &= \left(a,a\right)^T, \\
\ket{A_{1,1}} &= \left(a,-a\right)^T, \\
\ket{B_{0,0}} &= \left(c,d\right)^T, \\
\ket{B_{0,1}} &= \left(d,-c\right)^T, \\
\ket{B_{1,0}} &= \left(c,-d\right)^T, \\
\ket{B_{1,1}} &= \left(-d,-c\right)^T,
\end{aligned} 
\end{equation}
with $a = \frac{1}{\sqrt{2}}$, $c = \cos\left(\frac{\pi}{8}\right)$, and $d = \sin\left(\frac{\pi}{8}\right)$. 

For the CHSH case, the vector $v_i$ corresponds to $\ket{v_i}.$ The dimension of the canonical realization is $4$ with $d_1=d_2=2$. CHSH inequality satisfies Conditions A1 and A2,
which can be checked by choosing the vectors in \eqref{XO1MT} as follows:
\begin{align}
&a_0= |A_{0,0}\rangle ,~
a_1=a_2= |A_{0,1}\rangle ,  \\
&b_0= |B_{0,0}\rangle ,~
b_1= |B_{1,0}\rangle .
\end{align} 

Moreover, the local measurements for the CHSH case satisfy condition A3, A4 and C1 as well.  Thus, the CHSH case satisfies all the conditions for Theorem \ref{TH8MT}, which implies there exist isometries 
$V_A$ from $ {\cal H}_A\otimes  {\cal K}_A$ to $ {\cal H}_A'$
and 
$V_B$ from $ {\cal H}_B\otimes  {\cal K}_B$ to $ {\cal H}_B'$
such that
\begin{align}
V_A \otimes V_B |\psi\rangle \otimes |junk\rangle &=
|\psi'\rangle , \Label{ML1CHSH}\\
V_A \otimes V_B |v_i\rangle \otimes |junk\rangle &=|v_i'\rangle , \Label{ML2CHSH}
\end{align}
for $i \in {\cal I}$,
where
$|junk\rangle$ is a state on ${\cal K}_A\otimes  {\cal K}_B$.

Therefore, any two tensored realizations attaining the maximum quantum violation of the CHSH inequality are related via  local isometries.

\subsection{Mermin Self-Testing}

Here, we examine the case of Mermin's Bell inequality for three parties \cite{Mermin90}. As detailed in Appendix \ref{Mermin_APP}, the Bell witness of this inequality includes $16$ events.
Their graph of exclusivity, denoted $G_M$, is shown in Fig.~\ref{Mermin_fig}.

\begin{figure}[t] \label{Mermin_fig}

\centering
\includegraphics[width=0.58\textwidth]{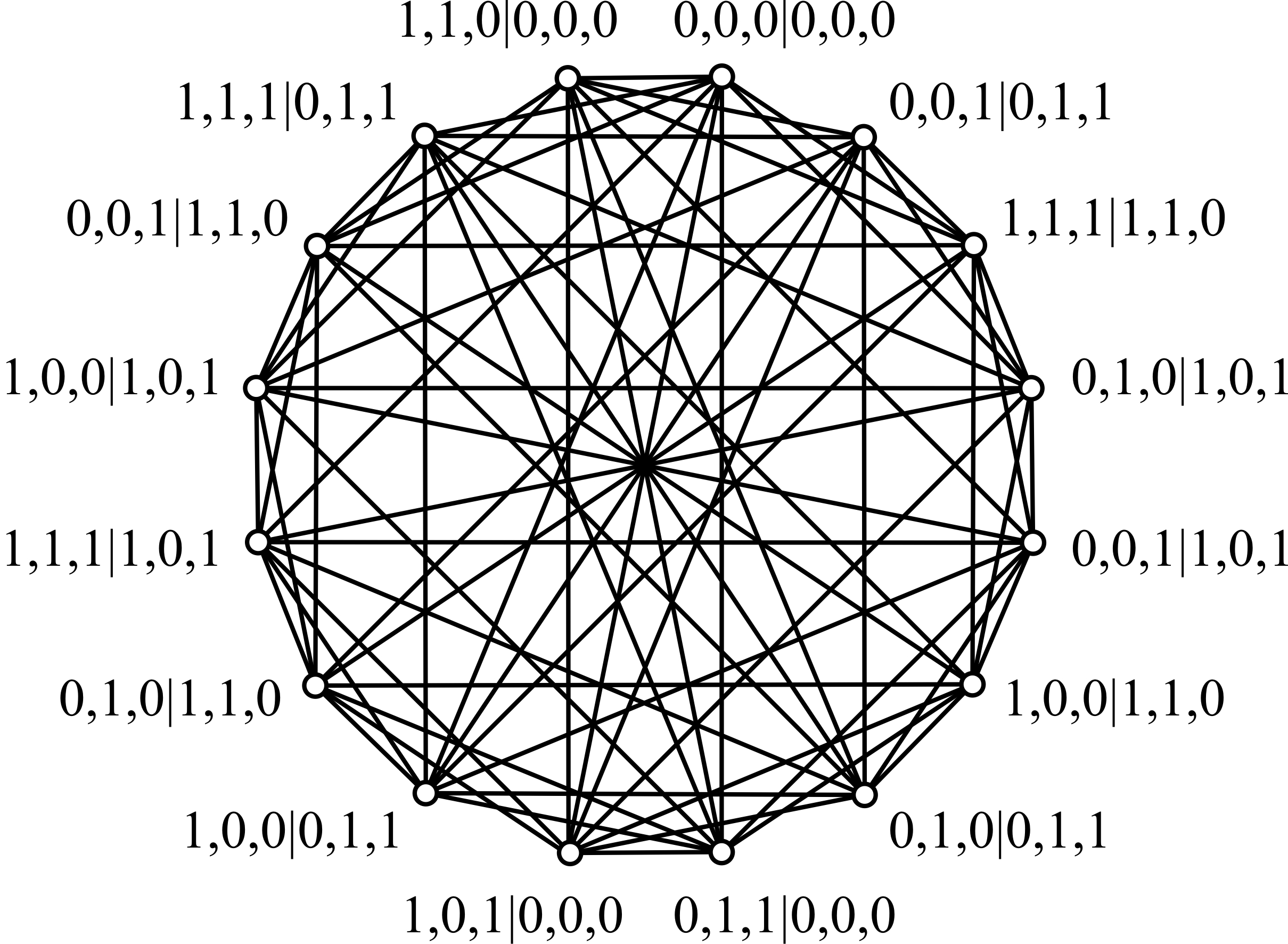}

\caption{Graph of exclusivity of the $16$ events in the Bell witness \eqref{S3} of the Mermin inequality. Here, $Z$ and $X$ are denoted $0$ and $1$, respectively, while $-1$ and $1$ are denoted $0$ and $1$, respectively. We will refer to this graph as $G_M$. It is the complement of Shrikhande graph~\cite{shrikhande1959uniqueness} }.

\centering
\end{figure}

The primal optimal for the SDP corresponding to the quantum violation of the Mermin inequality for three parties is given by 

\begin{equation}
P_{\text{Mermin}}=\left[\begin{array}{cc}
1 & a \cdot e_{16}^{T}\\
a \cdot e_{16} & \quad a\cdot I_{16} + b \cdot E_{\overline{G_M}}
\end{array}\right]\in\mathbb{R}^{\left(17 \times 17\right)},\label{primalMermin}
\end{equation}
where $a = 0.25$, $b = 0.125$, $e_{16}$ is the all one column vector of size $16$, and $E_{\overline{G_M}}$ is the adjacency matrix of the complement of $G_M$, and $I_{16}$ is the identity matrix of size $16$. The proof of the uniqueness of the primal optimal $P_{\text{Mermin}}$ is trivially similar to the CHSH case. The quantum state and measurement settings can be obtained via Gram decomposition of $P_{\text{Mermin}}$. A quantum realization is achieved with the three-qubit GHZ state $\ket{u_0} = \frac{1}{\sqrt{2}}\left(\ket{000} + \ket{111}\right)$ 
and the projective measurements
\begin{equation}
\begin{aligned}
\ket{u_1} &= \ket{Z} \otimes \ket{P} \otimes \ket{P}, \\
\ket{u_2} &= \ket{O} \otimes \ket{M} \otimes \ket{P}, \\
\ket{u_3} &= \ket{O} \otimes \ket{P} \otimes \ket{M}, \\
\ket{u_4} &= \ket{Z} \otimes \ket{M} \otimes \ket{M}, \\
\ket{u_5} &= \ket{P} \otimes \ket{Z} \otimes \ket{P}, \\
\ket{u_6} &= \ket{M} \otimes \ket{O} \otimes \ket{P}, \\
\ket{u_7} &= \ket{M} \otimes \ket{Z} \otimes \ket{M}, \\
\ket{u_8} &= \ket{P} \otimes \ket{O} \otimes \ket{M}, \\
\ket{u_9} &= \ket{P} \otimes \ket{P} \otimes \ket{Z}, \\
\ket{u_{10}} &= \ket{M} \otimes \ket{M} \otimes \ket{Z}, \\
\ket{u_{11}} &= \ket{M} \otimes \ket{P} \otimes \ket{O}, \\
\ket{u_{12}} &= \ket{P} \otimes \ket{M} \otimes \ket{O}, \\
\ket{u_{13}} &= \ket{O} \otimes \ket{O} \otimes \ket{O}, \\
\ket{u_{14}} &= \ket{Z} \otimes \ket{Z} \otimes \ket{O}, \\
\ket{u_{15}} &= \ket{Z} \otimes \ket{O} \otimes \ket{Z}, \\
\ket{u_{16}} &= \ket{O} \otimes \ket{Z} \otimes \ket{Z}, 
\end{aligned}   
\end{equation}
where $\ket{Z} = \ket{0}$, $\ket{O} = \ket{1}$, $\ket{P} =\frac{1}{\sqrt{2}} \left(\ket{0} + \ket{1} \right)$, and $\ket{M} =\frac{1}{\sqrt{2}} \left(\ket{0} - \ket{1} \right)$. This quantum realization achieves the quantum bound of the Bell witness (given by  Eq.~\eqref{S3} in Appendix \ref{Mermin_APP}), i.e., $4$, which is equal to the Lov\'asz theta number of $G_M$ (Fig.~\ref{Mermin_fig}). The local bound is $3$ and is equal to the independence number of $G_M$. We can check that local measurement settings for the tripartite Mermin case satisfy Conditions A5, A6, and A7 as follows.
In this example, $a_O,b_O,c_O$ means $|O\rangle$.
This notation is applied to $Z,P,M$.

We choose the subset ${\cal I}_A:=\{ O,P \}$.
Then, we have
\begin{align}
{\cal I}_{BC,O} &=\{ (O,O),(Z,Z),(M,P),(P,M)\},\label{MI1MT}\\
{\cal I}_{BC,P} &=\{ (Z,P),(P,Z),(O,M),(M,O)\}.\label{MI2MT}
\end{align}
Two elements $O,P \in {\cal I}_A$ are connected in the sense of the end of Definition \ref{D1MT} by choosing 
$\{i,j,k\}=\{ (P,Z,P),(O,Z,Z), (O,M,P) \}$.
Based on \eqref{MI1MT} and \eqref{MI2MT}, 
we choose the subsets ${\cal I}_B$, ${\cal I}_{C,Z}$, and ${\cal I}_{C,P}$ as
\begin{align}
{\cal I}_B:= \{Z,P\},~ {\cal I}_{C,Z}:=\{Z,P\},~
{\cal I}_{C,P}:=\{Z,M\}.
\end{align}
The subsets ${\cal I}_B$, ${\cal I}_{C,Z}$, and ${\cal I}_{C,P}$
satisfy conditions B1, B2, B3, and B4. The Mermin case satisfies Conditions A8 and A9 in addition to Conditions A5, A6, and A7. Thus, given the vectors $(\Pi_i |\psi'\rangle)_{i \in {\cal I}}$which  realize the optimal solution in the SDP \eqref{thetaSDP1}, there exist isometries 
$V_A$ from $ {\cal H}_A\otimes  {\cal K}_A$ to $ {\cal H}_A'$,
$V_B$ from $ {\cal H}_B\otimes  {\cal K}_B$ to $ {\cal H}_B'$,
and 
$V_C$ from $ {\cal H}_C\otimes  {\cal K}_C$ to $ {\cal H}_C'$
such that
\begin{align}
V_A \otimes V_B\otimes V_C |\psi\rangle \otimes |junk\rangle &=
|\psi'\rangle,\Label{ML5}\\
V_A \otimes V_B \otimes V_C|v_i\rangle \otimes |junk\rangle &=|v_i'\rangle , \Label{ML6}
\end{align}
for $i \in {\cal I}$,
where
$|junk\rangle$ is a state on ${\cal K}_A\otimes  {\cal K}_B\otimes  {\cal K}_C$.

Since $P_{\text{Mermin}}$ is a Gram matrix of vectors $\ket{u_0}, \ket{u_1},\ldots, \ket{u_{16}}$, rank of $P_{\text{Mermin}}$ is equal to the dimension of the span of $\ket{u_0}, \ket{u_1},\ldots, \ket{u_{16}}$. As it turns out, the rank of $P_{\text{Mermin}}$ is seven. Thus, a seven-dimensional configuration can achieve the maximal violation of the Mermin inequality. We append a seven-dimensional configuration corresponding to $P_{\text{Mermin}}$ in Appendix \ref{seven_m}. In the tripartite Bell scenario, one obtains maximal violation of the Mermin inequality using three qubits and thus dimension eight. This is purely because the seven dimensional state can not be realized as tensor product of three two-dimensional subsystems. Moreover, as one can expect, the dimension of the span of the measurement settings in the tensored case, i.e., dim(span$\left(\ket{u_0}, \ket{u_1},\ldots, \ket{u_{16}}\right)$) is still seven!

The graph $G_M$ is the complement of Shrikhande graph~\cite{shrikhande1959uniqueness}.

Since Shrikhande Graph is vertex transitive, it implies that $G_M$ is also vertex transitive. We observe that there is a unique behaviour in $\text{\rm QSTAB}\left(G_M\right)$ which achieves $\alpha^{\star}\left(G_M\right)$. Moreover, by vertex transitivity in the theta body, we also  observe that there is a unique behavior which achieves $\vartheta\left(G_M\right)$.

\subsection{Self-Testing Chained Bell Inequalities}\label{chained}

The chained Bell inequalities \cite{pearle70hidden,braunstein90wringling} are defined for the bipartite Bell scenario
with $N$ dichotomic measurements per party. In terms of correlations
between the observables of Alice and Bob, the chained Bell inequality
for $N$ settings is given by
\begin{equation}
I_{N}^{Bell}=\langle A_{1}B_{2}\rangle+\langle A_{3}B_{2}\rangle+\langle A_{3}B_{4}\rangle+\langle A_{5}B_{4}\rangle+\cdots+\langle A_{2N-1}B_{2N}\rangle-\langle A_{1}B_{2N}\rangle\leq_{LHV}2N-2.\label{eq:Chained_Bell_Ineq_1}
\end{equation}
Here, ``LHV'' indicates that the local hidden variable bound is
$2N-2.$ The observables $A_{i}$ and $B_{j}$, measured by Alice
and Bob, respectively, have outcomes $1$ or $-1.$
The correlation terms $\langle A_{i}B_{j}\rangle$ denote the expectation
value of the product of outcomes for $A_{i}$ and $B_{j}.$ The maximum
quantum value of $I_{N}^{Bell}$ is $2N\cos\left(\frac{\pi}{2N}\right)$. 

Suppose Alice measures $A_{x}$ on her particle and obtains $a$.
Similarly, assume Bob measures $B_{y}$ on his particle and obtains
$b.$ The probability for the aforementioned event is denoted by $P\left(a,b\vert x,y\right).$
We can use these probabilities to re-express the correlations as follows:
\begin{align}
\langle A_{i}B_{j}\rangle &=2P\left(1,1\vert i,j\right)+2P\left(-1,-1\vert i,j\right)-1,\label{eq:Prob1} \\
-\langle A_{i}B_{j}\rangle &=2P\left(1,-1\vert i,j\right)+2P\left(-1,1\vert i,j\right)-1.\label{eq:Prob2}
\end{align}
Using Eqs.~\eqref{eq:Prob1} and \eqref{eq:Prob2}, we can re-express
the inequality in equation \eqref{eq:Chained_Bell_Ineq_1} as
\begin{multline}
I_{N}^{\text{CSW}}=P\left(1,1\vert1,2\right)+P\left(-1,-1\vert1,2\right)+P\left(1,1\vert3,2\right)+P\left(-1,-1\vert3,2\right)+\cdots\\
+P\left(1,1\vert2N-1,2N\right)+P\left(-1,-1\vert2N-1,2N\right)+P\left(1,-1\vert1,2N\right)
+P\left(-1,1\vert1,2N\right)\leq_{LHV}2N-1.\label{eq:Chanined_Bell_CSW_1}
\end{multline}
The graph of exclusivity for the events in $I_{N}^{\text{CSW}}$ is $Ci_{4N}\left(1,2N\right)$
and is isomorphic to the M\"{o}bius ladder graph of order $4N$. The independence
number of $Ci_{4N}\left(1,2N\right)$ is $2N-1$. The Lov\'asz theta
number, however, remains unknown and has been conjectured~\cite{mateus} to be equal
to 

\begin{equation}
\vartheta\left(Ci_{4N}\left(1,2N\right)\right)=N\left[1+\cos\left(\frac{\pi}{2N}\right)\right].\label{eq:Lovasz_Mobius_Conjecture}
\end{equation}
Here, we prove that the above conjecture is correct by simple semidefinite
programming duality arguments. Moreover, we recover Bell self-testing
statements for the chained Bell inequalities for arbitrary $N.$ For
the purposes of the proof, we introduce the matrix

\begin{equation}
Z_{N}^{\star}=\left[\begin{array}{cc}
\frac{N}{l} & -e_{4N}^{T}\\
-e_{4N} & \quad l \cdot A_{C_{4N}} +\left[\begin{array}{cc}
I_{2N} & f \cdot I_{2N}\\
f \cdot I_{2N} & I_{2N}
\end{array}\right]
\end{array}\right]\in\mathbb{R}^{(4N+1) \times (4N+1)},\label{eq:Dual_Chained_Optimal}
\end{equation}
where $e_{4N}$ denotes the all-ones column vector of length $4N,$
$k=\cos\left(\frac{\pi}{2N}\right),$ $f=\frac{1-k}{1+k},$ $l = \frac{1}{1+k}$, $A_{C_{4N}}$
is the adjacency matrix of the cycle graph $C_{4N}$, and $I_{2N}$
is a $2N \times 2N$ identity matrix.
\begin{lemma}
\label{lem:Chained_positivity}$Z_{N}^{\star}\succcurlyeq0.$
\end{lemma}
\begin{proof}
Taking the Schur complement of $Z_{N}^{\star}$ with respect to its
top left entry, we have that
\begin{equation}
Z_{N}^{\star}\succcurlyeq0\iff M_{N}-\frac{l}{N}e_{4N}e_{4N}^{T}\succcurlyeq0,
\end{equation}
where $M_{N}=l \cdot A_{C_{4N}} +\left[\begin{array}{cc}
I_{2N} & f \cdot I_{2N}\\
f \cdot I_{2N} & I_{2N}
\end{array}\right]$. To prove that $Z_{N}^{\star}$ is positive semidefinite, it remains to
show that the eigenvalues of $M_{N}-\frac{l}{N}e_{4N}e_{4N}^{T}$
are non-negative. Notice that $e_{4N}$ is a common eigenvector of $M_{N}$ and $e_{4N}e_{4N}^{T}$ as both matrices have the property that the sum of the entries across a row is a constant. Hence, it suffices to compute all the eigenvalues of $M_{N}$. The eigenvalues of a circulant matrix are well characterised. 

\begin{fact}
The eigenvalues of the circulant matrix \begin{equation}
C=\left[\begin{array}{ccccc}c_{0} & c_{n-1} & \ldots & c_{2} & c_{1} \\ c_{1} & c_{0} & c_{n-1} & & c_{2} \\ \vdots & c_{1} & c_{0} & \ddots & \vdots \\ c_{n-2} & & \ddots & \ddots & c_{n-1} \\ c_{n-1} & c_{n-2} & \ldots & c_{1} & c_{0}\end{array}\right]
\end{equation}
are given by 
\begin{equation}
\lambda_{j}=c_{0}+c_{n-1} \omega^{j}+c_{n-2} \omega^{2 j}+\ldots+c_{1} \omega^{(n-1) j}, \quad j=0,1,\ldots, n-1,
\end{equation}
where $\omega = \exp (\frac{2\pi i}{n})$ is the $n^{th}$ root of unity.  
\end{fact}

Note that the matrix $M_N$ is a circulant matrix with $n = 4N$, $c_0 = 1,c_1 = l, c_{2N} = f, c_{n-1} = l$, and $c_i =0$ for $i \notin \{0,1,2N,n-1\}$. Therefore, its eigenvalues are given by $\lambda_j = 1+ l (\omega^j + \omega^{(n-1)j}) + f \omega^{2Nj} $, for $j = 0,1,\ldots n-1$. Simplifying this, we obtain 
\begin{equation}
\lambda_j=\begin{cases}
			1-f+2l \cos{\frac{\pi j}{2N}}, & \text{if $j$ is odd,}\\
            1+f+2l \cos{\frac{\pi j}{2N}}, & \text{if $j$ is even.}
		 \end{cases}
\end{equation}

When $j$ is even, the minimum eigenvalue is when $j = 2N$, for which 
\begin{equation}
\lambda_{2N} = 1+f - 2l = 1+ \frac{1-k}{1+k} -\frac{2}{1+k} =0.
\end{equation}
When $j$ is odd, the minimum eigenvalue is when $j = 2N-1$, for which 
\begin{equation}
\lambda_{2N-1} = 1-f + 2l \cos \left(\frac{(2N-1)\pi}{2N}\right) = 1-f - 2l \cos \left(\frac{\pi}{2N}\right) =1-f-2lk =  1 - \frac{1-k}{1+k} -\frac{2k}{1+k} =0.
\end{equation}

Finally, note that the eigenvalue of $M_N$ corresponding to the eigenvector $e_{4N}$ is $1 + 2l + f$. Whereas $\frac{l}{N}e_{4N}e_{4N}^{T}$ is a rank-1 matrix with eigenvector $e_{4N}$ with eigenvalue $\frac{l}{N} \times 4N = 4l$. Therefore, the eigenvalue of $ M_N - \frac{l}{N}e_{4N}e_{4N}^{T}$  corresponding to the eigenvector $e_{4N}$ is $1+2l +f -4l = 1+f-2l = 1 + \frac{1-k}{1+k} - \frac{2}{1+k} =0$. The rest of the eigenvalues of $ M_N - \frac{l}{N}e_{4N}e_{4N}^{T}$ are the same as those of $M_N$ and are non-negative as shown above. Hence, all the eigenvalues are non-negative. 
\end{proof}

\begin{claim}
The dual optimal corresponding to the optimization program \eqref{theta:dual} for $Ci_{4N}\left(1,2N\right)$
is $Z_{N}^{\star}$ (expression \ref{eq:Dual_Chained_Optimal}). 
\end{claim}
\begin{proof}
We need to show that
\begin{enumerate}
\item $Z_{N}^{*}$ is dual feasible for the program in \eqref{theta:dual}.
\item $Z_{N}^{\star}$ corresponds to dual optimal value.
\end{enumerate}

To show feasibility, we need to show that $Z_{N}^{\star}$ is of the form as in \eqref{theta:dual}, that is, $Z_{N}^{\star} =  tE_{00}+\sum_{i=1}^n (\lambda_i-1)E_{ii}-\sum_{i=1}^n\lambda_i E_{0i}+\sum_{i\sim j} \mu_{ij}E_{ij}$. This is indeed true for the following choice of values: $t = \frac{N}{l}$, $\lambda_i = 2$ for $i = 1,2,\ldots,4N$ and $\mu_{ij} = 2l$ whenever $i$ and $j$ share an edge in $C_{4N}$ and $\mu_{ij} = 2f$ for $|i-j| = 2N$. Finally, using Lemma \ref{lem:Chained_positivity}, we have $Z_{N}^{\star}\succcurlyeq0.$

Using the measurement settings for chained Bell inequalities in~\cite{AQBTC13}, one obtains  the  output of the primal SDP \eqref{theta:primal} for  $I_{N}^{CSW}$ equal to  $N\left[1+\cos\left(\frac{\pi}{2N}\right)\right]$. Strong duality for the SDP in  \eqref{theta:primal} implies that $Z_{N}^{\star}$ corresponds to dual optimal value.
\end{proof}
The proof of the uniqueness of the primal optimal is similar to the proof corresponding to $n$-cycle graphs in ~\cite{BRVWCK19}. Chained Bell Inequalities satisfies Conditions A1 and A2, 
which can be checked by choosing the vectors in \eqref{XO1} as follows:
\begin{align}
&a_0= |A_{1}=1\rangle ,~
a_1= |A_{3}=1\rangle ,~
a_2= |A_{2N-1}=-1\rangle , \\
&b_0= | B_2=1 \rangle ,~
b_1= |B_{2N}=-1\rangle .
\end{align}
Here, $|A_{1}=1\rangle $ expresses the eigenvector of $A_1$ with eigenvalue $1$.
This notation is applied to other observables. Since the optimal maximizer given in \eqref{NA1MT} satisfies conditions A1 and A2
and
the vectors $(\Pi_i |\psi'\rangle)_{i \in {\cal I}}$ realize the optimal solution in the SDP \eqref{thetaSDP1}.
In addition, we the ranks of the projections $\Pi_{i_A}^A$ and $\Pi_{i_B}^B$ 
are assumed to be one.
Thus, there exist isometries 
$V_A: {\cal H}_A \to {\cal H}_A'$
and $V_B: {\cal H}_B\to {\cal H}_B'$ such that
\begin{align}
V_A \otimes V_B|\psi \rangle&= |\psi'\rangle , \\
V_A \otimes V_B|v_i \rangle&= |v_i'\rangle ,
\end{align}
for $i \in {\cal I}$. This completes the proof of self-testability for the chained Bell inequalities for rank-one projectors.

Since $Z_{N}^{\star}$ corresponds to dual optimal value, we have $\vartheta\left(Ci_{4N}\left(1,2N\right)\right)=N\left[1+\cos\left(\frac{\pi}{2N}\right)\right]$ as conjectured in~\cite{mateus}.

\subsection{Abner Shimony Self-Testing}

The Abner Shimony (AS) Bell inequalities~\cite{gisin2009bell} refer to a bipartite Bell scenario with an even number $n$ of measurement settings per party. Each measurement has two outcomes. It can be written as
\begin{align}\label{eq:as_ori}
  AS_n = \sum_{i+j< n} \langle A_i B_j \rangle - \sum_{i+j=n} \min\{i-1,j-1\} \langle A_i B_j\rangle\overset{LHV}{\le} \frac{n(n+2)}{4}.
\end{align}
By taking into account that
\begin{align}
  \langle A_i B_j \rangle &= 2[P(0,0|i,j)+P(1,1|i,j)]-1,\\ - \langle A_i B_j\rangle &= 2[P(0,1|i,j)+P(1,0|i,j)]-1,
\end{align}
Eq.~\eqref{eq:as_ori} can be rewritten as
\begin{align}
  AS^{c}_n =& \sum_{i+j< n} [P(0,0|i,j)+P(1,1|i,j)]\nonumber\\ &+ \sum_{i+j=n} \min\{i-1,j-1\} [P(0,1|i,j)+P(1,0|i,j)] \overset{LHV}{\le} \frac{n(n+1)}{2}. 
\end{align}
For example, for the case $n=4$,
\begin{align}
\label{ASC4}
  AS^c_4 = & P(0,0|0,0)+P(1,1|0,0)+P(0,0|0,1)+P(1,1|0,1)+P(0,0|0,2)+P(1,1|0,2)\nonumber\\ &+P(0,0|0,3)+P(1,1|0,3)+P(0,0|1,0)+P(1,1|1,0)+P(0,0|1,1)+P(1,1|1,1)\nonumber\\ &+P(0,0|1,2)+P(1,1|1,2)+P(0,0|1,3)+P(1,1|1,3)+P(0,0|2,0)+P(1,1|2,0)\nonumber\\ &+P(0,0|2,1)+P(1,1|2,1)+2[P(0,0|2,2)+P(1,1|2,2)]+P(0,0|3,0)+P(1,1|3,0)\nonumber\\ &+P(0,0|3,1)+P(1,1|3,1).
\end{align}
The (vertex-weighted) graph of exclusivity of the $26$ events in Eq.~(\ref{ASC4}) is shown in Fig.~\ref{AS4_ex} and has $\alpha(G,w) = 10$, $\vartheta(G,w) = 7 + \frac{5 \sqrt{6} }{3}$, and $\alpha^*(G,w) = 14$. Notice that the vertex weight is $2$ for events $[0,0|2,2]$ and $[1,1|2,2]$ and $1$ otherwise. 
\begin{equation}
\begin{aligned}
&|w_1\rangle = |A_0 \rangle \otimes \vert A_0\rangle,\\
&|w_2\rangle = |B_0 \rangle \otimes \vert B_0\rangle,\\
&|w_3\rangle = |A_0 \rangle \otimes \vert A_1\rangle,\\
&|w_4\rangle = |B_0 \rangle \otimes \vert B_1\rangle,\\
&|w_5\rangle = |A_0 \rangle \otimes \vert A_2\rangle,\\
&|w_6\rangle = |B_0 \rangle \otimes \vert B_2\rangle,\\
&|w_7\rangle = |A_0 \rangle \otimes \vert A_3\rangle,\\
&|w_8\rangle = |B_0 \rangle \otimes \vert B_3\rangle,\\
&|w_9\rangle = |A_1 \rangle \otimes \vert A_0\rangle,\\
&|w_{10}\rangle = |B_1 \rangle \otimes \vert B_0\rangle,\\
&|w_{11}\rangle = |A_1 \rangle \otimes \vert A_1\rangle,\\
&|w_{12}\rangle = |B_1 \rangle \otimes \vert B_1\rangle,\\
&|w_{13}\rangle = |A_1 \rangle \otimes \vert A_2\rangle,\\
&|w_{14}\rangle = |B_1 \rangle \otimes \vert B_2\rangle,\\
&|w_{15}\rangle = |A_1 \rangle \otimes \vert A_3\rangle,\\
&|w_{16}\rangle = |B_1 \rangle \otimes \vert B_3\rangle,\\
&|w_{17}\rangle = |A_2 \rangle \otimes \vert A_0\rangle,\\
&|w_{18}\rangle = |B_2 \rangle \otimes \vert B_0\rangle,\\
&|w_{19}\rangle = |A_2 \rangle \otimes \vert A_1\rangle,\\
&|w_{20}\rangle = |B_2 \rangle \otimes \vert B_1\rangle,\\
&|w_{21}\rangle = |A_2 \rangle \otimes \vert A_2\rangle,\\
&|w_{22}\rangle = |B_2 \rangle \otimes \vert B_2\rangle,\\
&|w_{23}\rangle = |A_3 \rangle \otimes \vert A_0\rangle,\\
&|w_{24}\rangle = |B_3 \rangle \otimes \vert B_0\rangle,\\
&|w_{25}\rangle = |A_3 \rangle \otimes \vert A_1\rangle,\\
&|w_{26}\rangle = |B_3 \rangle \otimes \vert B_1\rangle,
\end{aligned}    
\end{equation}

\begin{figure}[t] \label{AS4_ex}
\centering
\includegraphics[width=0.52\textwidth]{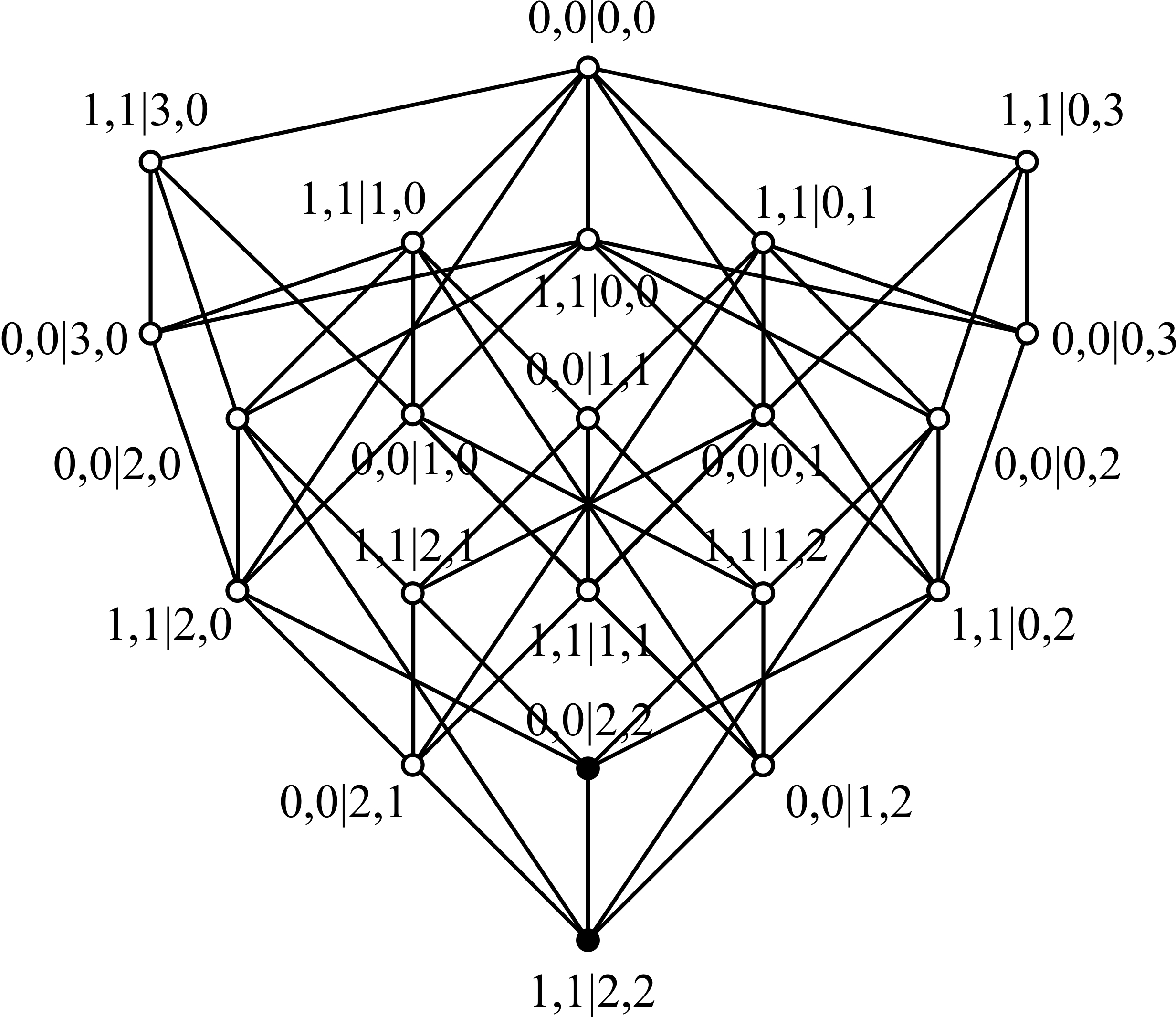}
\caption{Vertex-weighted graph of exclusivity for the events in the Bell inequality AS4 in Eq.~(\ref{ASC4}). There are $26$ events. Black nodes represent vertices with weight $2$ in Eq.~(\ref{ASC4}) and white nodes represent vertices with weight $1$.}
\centering
\end{figure}

The violation of the Bell inequality $AS^c_4$ can achieve $\vartheta(G,w)$ by choosing as initial state
\begin{equation}
\label{stateAS4}
  |s\rangle = \cos t (|00\rangle-|11\rangle) + \sin t (|01\rangle-|10\rangle),
\end{equation}
and as local measurements 
\begin{equation}
    A_i = |m(\alpha_i)\rangle, B_i = |m(\pi/2 + \alpha_i)\rangle,
\end{equation}
with $i=0,1,2,3$, $m(\alpha) = \cos\alpha|0\rangle + \sin\alpha|1\rangle$, and
\begin{align}\label{eq:as4setting}
  \alpha_0 = 0,\,\,\,\alpha_1 = \arcsin\left(\frac{1}{ \sqrt{6} }\right),\,\,\,
  \alpha_2 = \frac{1}{2} \left(\pi -\arctan\left(\sqrt{\frac{5 \sqrt{145}}{8}+\frac{77}{8}}\right)\right), \nonumber \\
  \alpha_3 = \frac{1}{2} \left(\pi -\arctan\left(48 \sqrt{\frac{2}{275 \sqrt{145}+3317}}\right)\right),\,\,\,
  t = \frac{1}{8}\left(\alpha_2 + 2\alpha_4 - \frac{\pi}{2} \right).
\end{align}

The primal optimal for the SDP corresponding to the quantum violation of $AS^c_4$ can be obtained by the state and measurement directions given in Eqs.~\eqref{stateAS4}--\eqref{eq:as4setting}. Here, we omit its expression, as it is lengthy and complex. The proof of the uniqueness of the primal optimal is similar as in previous cases. 

The local projective measurements satisfy Conditions A1 and A2, 
which can be checked by choosing the vectors in \eqref{XO1} as follows:
\begin{align}
&a_0= |A_{2}=0\rangle ,~
a_1=a_2= |A_{3}=0\rangle ,  \\
&b_0= | B_0=0 \rangle ,~
b_1= |B_{1}=0\rangle .
\end{align}

Here, $|A_{1}=1\rangle $ expresses the eigenvector of $A_1$ with eigenvalue $1$.
This notation is applied to other observables. Since the optimal maximizer given in \eqref{NA1MT} satisfies conditions A1 and A2
and
the vectors $(\Pi_i |\psi'\rangle)_{i \in {\cal I}}$ realize the optimal solution in the SDP \eqref{thetaSDP1}.
In addition, we the ranks of the projections $\Pi_{i_A}^A$ and $\Pi_{i_B}^B$ 
are assumed to be one.
Thus, there exist isometries 
$V_A: {\cal H}_A \to {\cal H}_A'$
and $V_B: {\cal H}_B\to {\cal H}_B'$ such that
\begin{align}
V_A \otimes V_B|\psi \rangle&= |\psi'\rangle , \\
V_A \otimes V_B|v_i \rangle&= |v_i'\rangle ,
\end{align}
for $i \in {\cal I}$. This completes the proof of self-testability for the $AS^c_4$ Bell inequality for rank-one projectors.

\section{Summary and Open Problems}\label{discuss}

In this work, we introduced a graph-theoretic approach to self-testing in Bell scenarios, combining ideas from graph theory and semidefinite programming. The motivation was the observation that the set of quantum correlations for a Bell scenario is, in general, difficult to characterize while, using ideas from ~\cite{CSW}, one can provide an easy to characterize single SDP-based relaxation of this set. By proving self-testing for the maximizer of a Bell inequality with respect to the aforementioned set, we furnish self-testing for the set of quantum correlations for the underlying Bell scenario. 

Our method requires that the quantum bound of the Bell inequality is equal to the Lov\'asz theta number of the vertex-weighted graph of exclusivity of the events appearing in the Bell witness, when written as a positive linear combination of probabilities of events. As we have seen, this is frequently the case. Our other assumptions involve some particular relation among the local projective measurements involved in the scenario as mentioned in Theorems \ref{TH7MT}, \ref{TH8MT}, \ref{TH9MT} and \ref{TH10MT}. In future, it would be interesting to simplify our assumptions involving relation among local projective measurements.

We applied our techniques to the CHSH, chained and three-party Mermin Bell inequalities. For CHSH and the trpartite Mermin case, we recovered self-testing results for projectors of arbitrary rank. For chained Bell inequalities, our self-testing statements hold for rank-one projectors. For the Mermin three-party case, the primal optimal matrix's rank is seven, indicating that the self-testing preparation dimension can be seven. However, in the Bell scenario, the underlying dimension has to be eight due to the tensor structure.
We also applied our method to the previously not-known case of AS inequalities and provided a self-testing statement for the case of rank-one projectors.

While delivering the self-testing statement for the chained-Bell inequality via our graph-theoretic framework, we also obtained a closed-form expression for the Lov\'asz theta number for M\"{o}bius ladder graphs. Our closed-form expression matches with the conjecture of~\cite{mateus}. 

Our methods belong in the intersection region of graph theory, Bell non-locality, and contextuality. Our results provide further motivation to study Bell self-testing via the graph-theoretic framework in the future. Furthermore, we believe that techniques such as ours could be used in the future to study open problems in graph theory taking advantage of ideas from quantum theory.

A natural next step in our program would be to generalize our result for scenarios with noise. In other words, a graph-theoretic approach to robust Bell self-testing. 

The graph-theoretic approach has been employed to study self-testing in Bell scenarios and in contextuality scenarios with sequential measurements. It will be interesting to see if the techniques based on graph theory could be also useful for self-testing in prepare and measure scenarios. In future, it would be interesting to extend our self-testing statements for chained Bell and AS inequalities for arbitrary rank projectors.

\medskip
{\noindent {\em Acknowledgements---}} AC was supported by Project Qdisc (Project No.\ US-15097), with FEDER funds, MINECO Project No.\ FIS2017-89609-P, with FEDER funds, and QuantERA grant SECRET, by MINECO (Project No.\ PCI2019-111885-2). KB and LCK thank National Research Foundation and the Ministry of Education, Singapore for the financial support. ZPX was supported by the Alexander von Humboldt Foundation. MH is supported by Guangdong Provincial Key Laboratory under Grant 2019B121203002. MR is supported by MEXT Quantum Leap Flagship Program (MEXT Q-LEAP) Grant Number JPMXS0120319794.

\bibliography{CHSH_ST.bib}
\bibliographystyle{alpha}

\appendix

\section{Graph Theory Basics}\label{GTB}

A graph $G = (V,E)$ consists of a set of vertices $V$ and edges $E$ \cite{Book}. Two vertices are {\em adjacent} if they share an edge between them. The complement graph $\bar{G}$ has same vertices as $G$, but its edge set is complement of the set $E$. A {\em clique} of a graph  is set of pairwise adjacent vertices.  The complement of clique is a set of vertices that are pairwise non-adjacent. A natural generalization of graph is hypergraph with  generalized edges connecting more than two vertices. These generalized edges are called hyperedges.

\begin{definition} (Cyclic graph) Given a graph with $n$ vertices such that every $i$th vertex of the graph is connected to $(i+1) \text{mod }n \, ^{th}$ vertex and $(i-1) \text{mod } n \, ^{th}$ vertex is called cyclic graph and denoted as $C_n$.
\end{definition}

An elegant generalization of the concept of cyclic graph is the concept of circulant graph, which is defined below.
\begin{definition} (Circulant graph) \label{circulantgraph}
Given a list $[L]$ of integers, a graph with $n$ vertices where $i \, ^{th}$ vertex is connected to $(i+l) \text{mod } n \, ^{th}$ and $(i-l) \text{ mod } n \, ^{th}$ vertices for all $l \in [L] $, is referred as circulant graph $Ci_n[L]$.  $Ci_n[1]$ graphs are called cyclic graphs.
\end{definition}

\begin{definition} (Orthonormal representation of a graph \cite{lovasz1979shannon}) An orthonormal representation of a graph is an assignment of unit vectors $\ket{v_i} \in \mathbb{R}^d$ to every vertex $i \in V$ such that 
\beq
\braket{v_i}{v_j}=0, \forall i,j \notin E. 
\enq 
We will use the notation ${\rm OR}(G)$ to represent the orthonormal representation of $G$. 
\end{definition}
\begin{definition} (Stable set) Stable set is a set of vertices of a graph such that no two vertices which lie in it share an edge.
\end{definition}

\begin{definition} (Independence number)
Independence number of a graph is the cardinality of the largest stable set of the graph. We will denote it by $\alpha(G).$

\end{definition}

\begin{definition} (Convex hull) Convex hull of a set $A$ is the smallest convex set containing $A$.
\end{definition}

\begin{definition} (Incidence vector) An Incidence vector of a set $B \subset A$ is a vector $P \in \mathbb{R}^{\vert A\vert}_+$ such that for every $i \in A$,
\begin{equation}
  P_i= 
\begin{cases}
   1 & \text{if }  i \in B,\\
    0              & \text{otherwise.}
\end{cases}  
\end{equation}
\end{definition}

\begin{definition} (Stable set polytope) The convex hull of all the incidence vectors of stable sets of graph $G$ is called stable set polytope of graph. It is denoted by $\text{\rm STAB}(G)$.
\end{definition}

\begin{definition} \label{Theta1}(Theta body) Let $\{\ket{v}_i\}$ corresponds to the orthonormal representation of $\bar{G}$. Given a unit vector $\ket{\phi}= (1,0,0,\cdots,0) \in \mathbb{R}^d$ with only first co-ordinate $1$ and rest $0$, the Theta body of graph $G$ is defined as
\begin{equation}
    \text{\rm TH}(G)=\{P \in \mathbb{R}^{\vert V \vert} : P_i =  \vert \braket{\psi}{v_i}\vert^2\}.
\end{equation}
\end{definition}

\begin{definition} \label{lovaszdefinition} (Lov\'asz theta number \cite{lovasz1979shannon})
The Lov\'asz theta number $\vartheta(G)$ of a graph $G$ is defined as follows:
\beq
\vartheta(G)= \max_{\ket{\phi}, \{\ket{v_i}\}}\sum_{i} |\braket{\phi}{v_i}|^2, \nonumber
\enq
where $\ket{\phi}$ is a unit vector and $\{\ket{v_i}\}$ is an orthonormal representation of the graph $G$. $\ket{\phi}$ is also known as handle.
\end{definition}
\begin{definition} (Fractional stable set polytope) The fractional stable set polytope is given by
\begin{equation}
\text{\rm QSTAB}(G)=\left\{P \in \mathbb{R}^{\vert V \vert_+} : \sum_{i \in C}P_i \le 1 \text{ for every clique } C \text{ of graph } G \right\}.
\end{equation}
\end{definition}
\begin{definition} (Fractional packing number) The fractional packing of a graph $G$ is the value of the following linear program:
\begin{equation}
    \al^*(G)=\max\left\{\sum_{i=1}^n x_i: x\in {\rm QSTAB}(G)\right\}.
\end{equation}
\end{definition}

\begin{definition} (Gram matrix and Gram decomposition) Given a set of vectors $v_1, v_2,\ldots, v_k$ in an inner product space, the corresponding Gram matrix is a Hermitian matrix $X$, defined via their inner products such that $X_{i,j} = \langle v_i, v_j\rangle$ for $i,j \in \{1,2,\ldots,n\}.$ It is important to note that $\text{rank } X = \text{dim span}\left(v_1, v_2,\ldots, v_k \right).$ Decomposing Gram matrix $X$ such that $X = AA^{\dagger}$ is called Gram decomposition of $X$. The rows of $A$ are related to $v_i$ up to isometry.
\end{definition}

\begin{definition} (Vertex-transitive graph) A graph $G = \left(V,E\right)$ is called vertex transitive, if given any two vertices $v_1, v_2 \in V $there exists an automorphism $h: V \rightarrow V$ such that $h\left(v_1\right) = v_2$.
\end{definition}

\begin{fact} \cite{lovasz1979shannon}
For a given graph $G,$ $\alpha(G) \leq \vartheta(G)\leq \al^*(G).$
\end{fact}
It is worthwhile to note that $\text{\rm STAB}(G) \subseteq$ TH(G)$\subseteq$ \text{\rm QSTAB}(G) \cite{knuth1994sandwich}.
An alternate formulation of theta body of a graph $G=([n],E)$ is given by: 
\begin{equation}
{\rm TH}(G)=\{x\in \R^n_+: \exists X\in \mathbb{S}^{1+n}_+, \ X_{00}=1, \  X_{ii}=X_{0i}, \ X_{ij}=0, \forall ij\in E\}.
\end{equation}
\begin{lemma}~\cite{ray2021graph}\label{csdcever} We have that $x\in {\rm TH}(G)$ iff there exist unit vectors $d,w_1,\ldots,w_n$ such that 
\be
x_i=\la d,w_i\ra^2, \forall i\in [n] \text{ and } \la w_i, w_j\ra=0, \text{ for } ij\in E.
\ee
\end{lemma}

\section{Mermin inequality} \label{Mermin_APP}

Mermin's Bell inequality \cite{Mermin90} refers to a $n$-partite Bell scenario (with $n \ge 3$ odd; there is also a version for $n$ even \cite{Ardehali92}, but we won't consider it here). The interest of this Bell inequality is based on the fact that the Bell operator
\begin{equation}
S_n = \frac{1}{2 i} \left[\bigotimes_{j=1}^n (\sigma_x^{(j)}+i \sigma_z^{(j)}) - \bigotimes_{j=1}^n (\sigma_x^{(j)}-i \sigma_z^{(j)})\right],
\end{equation}
where $\sigma_x^{(j)}$ is the Pauli matrix $x$ for qubit $j$, has an eigenstate with eigenvalue $2^{(n-1)}$. In contrast, for local hidden-variable (LHV) and non-contextual hidden-variable (NCHV) theories,
\begin{equation}
\langle S_n \rangle \overset{\scriptscriptstyle{\mathrm{LHV, NCHV}}}{\le} 2^{(n-1)/2}.
\end{equation}
For example,
\begin{eqnarray}
S_3 &=& \sigma_z^{(1)} \otimes \sigma_x^{(2)} \otimes \sigma_x^{(3)}+
\sigma_x^{(1)} \otimes \sigma_z^{(2)} \otimes \sigma_x^{(3)}\\&	&+
\sigma_x^{(1)} \otimes \sigma_x^{(2)} \otimes \sigma_z^{(3)}-
\sigma_z^{(1)} \otimes \sigma_z^{(2)} \otimes \sigma_z^{(3)}.
\end{eqnarray}
Therefore, we can write (using obvious notation),
\begin{equation}
\langle S_3 \rangle = \langle ZXX \rangle+\langle XZX \rangle+\langle XXZ \rangle-\langle ZZZ \rangle.
\end{equation}
Then, by taking into account that
 \begin{widetext}
\begin{eqnarray*}
 \langle ZXX \rangle &=& P(Z=X=X=1) + P(Z=X=-X=-1) + P(Z=-X=X=-1) + P(-Z=X=X=-1)
 \\& & -P(Z=X=X=-1) - P(Z=X=-X=1) - P(Z=-X=X=1) - P(-Z=X=X=1) 
  \\&=&2 \left[ P(Z=X=X=1) + P(Z=X=-X=-1) + P(Z=-X=X=-1) + P(-Z=X=X=-1) \right] -1,\\
 \langle XZX \rangle &=& P(X=Z=X=1) + P(X=Z=-X=-1) + P(X=-Z=X=-1) + P(-X=Z=X=-1)
 \\& & -P(X=Z=X=-1) - P(X=Z=-X=1) - P(X=-Z=X=1) - P(-X=Z=X=1)
 \\&=&2 \left[ P(X=Z=X=1) + P(X=Z=-X=-1) + P(X=-Z=X=-1) + P(-X=Z=X=-1) \right] -1,\\
 \langle XXZ \rangle &=& P(X=X=Z=1) + P(X=X=-Z=-1) + P(X=-X=Z=-1) + P(-X=X=Z=-1) 
 \\& & -P(X=X=Z=-1) - P(X=X=-Z=1) - P(X=-X=Z=1) - P(-X=X=Z=1) 
 \\&=&2 \left[ P(X=X=Z=1) + P(X=X=-Z=-1) + P(X=-X=Z=-1) + P(-X=X=Z=-1) \right] -1,\\
 -\langle ZZZ \rangle &=& P(Z=Z=Z=-1) + P(Z=Z=-Z=1) + P(Z=-Z=Z=1) + P(-Z=Z=Z=1) 
 \\& & -P(Z=Z=Z=1) - P(Z=Z=-Z=-1) - P(Z=-Z=Z=-1) - P(-Z=Z=Z=-1)\\
&=& 2 \left[ P(Z=Z=Z=-1) + P(Z=Z=-Z=1) + P(Z=-Z=Z=1) + P(-Z=Z=Z=1) \right] -1,
\end{eqnarray*}
\end{widetext}
we can rewrite $\langle S_3 \rangle$ as a sum of the probabilities of 16~events. That is,
 \begin{widetext}
	\begin{equation}
	\label{S3}
	\begin{aligned}
		\langle S_3 \rangle =& 2 \left[ P(Z=X=X=1) + P(Z=X=-X=-1) + P(Z=-X=X=-1) + P(-Z=X=X=-1) \right. \\
	&+	P(X=Z=X=1) + P(X=Z=-X=-1) + P(X=-Z=X=-1) + P(-X=Z=X=-1) \\
&+	P(X=X=Z=1) + P(X=X=-Z=-1) + P(X=-X=Z=-1) + P(-X=X=Z=-1) \\
&	\left. +P(Z=Z=Z=-1) + P(Z=Z=-Z=1) + P(Z=-Z=Z=1) + P(-Z=Z=Z=1) \right] -4.
\end{aligned}
	\end{equation}
\end{widetext}
The graph of exclusivity of these 16~events is the complement of Shrikhande graph~\cite{shrikhande1959uniqueness}.

This graph, shown in Fig.~\ref{Mermin_fig}, has $\alpha = 3$ and $\vartheta = \alpha^* =4$.
Similarly, one can obtain the graph corresponding to any $\langle S_n \rangle$.

\section{Seven dimensional configuration for the Mermin case} \label{seven_m}

We have obtained numerically (rounded up to three digits after decimal) the following seven dimensional configuration achieving the Lov\'az theta number of the graph in Fig.~\ref{Mermin_fig}. 
\begin{equation}
\begin{aligned}
\ket{u_0} &= \left(1,0,0,0,0,0,0\right)^T,\\
\ket{u_1} &=   \left(0.25,	-0.113,	-0.241,	0.284,	0.088,	0.166,	-0.029\right)^T,\\
\ket{u_2} &=  \left( 0.25,	-0.110,	-0.251,	-0.120,	0.247,	-0.021,	-0.191   \right)^T,\\
\ket{u_3} &=  \left(0.25,	-0.292,	0.079,	0.151,	0.075,	-0.051,	-0.255    \right)^T,\\
\ket{u_4} &=  \left( 0.25,	0.182,	-0.087,	0.003,	0.311,	0.215,	0.059  \right)^T,\\
\ket{u_5} &=  \left( 0.25,	-0.226,	0.069,	0.104,	-0.227,	0.262,	-0.021	   \right)^T,\\
\ket{u_6} &=  \left( 0.25,	0.223,	-0.059,	0.300,	0.068,	-0.075,	0.184,   \right)^T,\\
\ket{u_7} &=  \left(0.25,	-0.004,	-0.232,	0.130,	-0.298,	0.001,	0.167    \right)^T,\\
\ket{u_8} &= \left( 0.25,	-0.247,	0.049,	-0.152,	0.140,	-0.278,	0.059   \right)^T,\\
\ket{u_9} &=  \left(0.25,	0.251,	-0.059,	-0.252,	0.019,	0.091,	-0.222	    \right)^T,\\
\ket{u_{10}} &= \left(0.25,	0,	-0.242,	-0.274,	-0.139,	-0.186,	0.004    \right)^T,\\
\ket{u_{11}} &= \left(0.25,	0.069,	0.271,	0.019,	-0.154,	0.062,	-0.285    \right)^T,\\
\ket{u_{12}} &=  \left(0.25,	0.044,	0.261,	0.167,	0.054,	-0.291,	-0.042    \right)^T,\\
\ket{u_{13}} &=  \left(0.25,	0.069,	0.223,	-0.178,	-0.004,	0.312,	0.067    \right)^T,\\
\ket{u_{14}} &=  \left(0.25,	0.045,	0.212,	-0.030,	0.204,	-0.042,	0.310    \right)^T,\\
\ket{u_{15}} &=  \left( 0.25,	-0.182,	0.039,	-0.200,	-0.161,	0.035,	0.293	   \right)^T,\\
\ket{u_{16}} &= \left(0.25,	0.291,	-0.031,	0.046,	-0.225,	-0.199,	-0.097	    \right)^T.
\end{aligned}   
\end{equation}

\section{Proofs of Self-testing} \label{ST_Proof}
We consider two types of sets of indexes ${\cal I}$ and ${\cal I}_0={\cal I} \cup \{0\}$.
We consider 
the matrix $X_{ij}:=\langle \psi | \Pi_j \Pi_i |\psi\rangle$,
where $\Pi_i$ is a projection and 
$\Pi_0$ is the identity operator.
We set $n:=|{\cal I}|$.
Then, we assume that the following SDP has the unique solution.
\begin{equation}\label{thetaSDP}
\begin{aligned} 
\vartheta(\gexcl,w) = \max & \  \sum_{i \in {\cal I}} w_i { X}_{ii} \\
{\rm  s.t.}   &  \  { X}_{ii}={ X}_{0i}, \ \forall i\in [n],\\
  & \ { X}_{ij}=0,\ \forall i\sim j,\\
& \ X_{00}=1,\  X\in \mathbb{S}^{1+n}_+.
\end{aligned}
\end{equation}

\subsection{Bipartite case}
We assume that the unique optimal maximizer 
$X^*=(X_{ij})$ is given by $\eta_i \eta_j\langle v_j, v_i\rangle$ with the following;
For $i=(i_A,i_B)\in {\cal I}$, 
\begin{align} v_{i} = a_{i_A} \otimes b_{i_B}, 
\label{NA1}
\end{align}
where $a_{i_A} \in {\cal H}_A=\mathbb{C}^{d_A}$, 
$b_{i_B} \in {\cal H}_B=\mathbb{C}^{d_B}$.
Also, for simplicity, $a_{i_A} $ and $ b_{i_B}$ are assumed to be normalized
and $\eta_i>0$.

Now, we consider a state $|\psi'\rangle $ on ${\cal H}_A'\otimes {\cal H}_B'$,
and 
projections $\Pi_{i_A}^A$ and $\Pi_{i_B}^B$ on ${\cal H}_A'$ and ${\cal H}_B'$.
Here, 
when $i_A=i_A'$ ($i_B=i_B'$) for $i \neq i'$,
$\Pi_{i_A}^A=\Pi_{i_A'}^A$ ($\Pi_{i_B}^B=\Pi_{i_B'}^B$).
Then, we define the projection 
$\Pi_i:=\Pi_{i_A}^A\otimes \Pi_{i_B}^B$,

In the following, 
we discuss how 
the state $ |\psi'\rangle$ is locally converted to 
$ |\psi\rangle$ when
the vectors $\Pi_i |\psi'\rangle$ realize the optimal solution in the SDP \eqref{thetaSDP}.
We define $|v_i'\rangle:= \eta_i^{-1} \Pi_{i}|\psi'\rangle$.

\subsubsection{Rank-one case}
First, we consider the case that the ranks of the projections $\Pi_{i_A}^A$ and $\Pi_{i_B}^B$ 
are one.
We introduce the following conditions.
\begin{description}
\item[A1]
The set $\{v_i\}_{i\in {\cal I}_0}$ of vectors 
span the vector space ${\cal H}_A\otimes {\cal H}_B$.

\item[A2]
There exist
a subset ${\cal I}_B$ of indexes of the space ${\cal H}_B$ with
$|{\cal I}_B|=d_B=\dim {\cal H}_B$
and 
$d_B$ sets $\{ {\cal I}_{A,i_B} \}_{i_B \in {\cal I}_B}$ of indexes of the space 
${\cal H}_A$ 
$|{\cal I}_{A,i_B}|=d_A=\dim {\cal H}_A$
to satisfy the following conditions B1-B4.
\end{description}

\begin{description}
\item[B1]
$\cup_{i_B \in {\cal I}_B} {\cal I}_{A,i_B}\times \{i_B\}
\subset {\cal I}$.
\item[B2]
$\{ b_{i_B} \}_{i_B \in {\cal I}_B}$ spans the space ${\cal H}_B$.
\item[B3]
$\{ a_{i_A} \}_{i_A \in {\cal I}_{A, i_B}}$ spans the space ${\cal H}_A$
for any $i_B \in {\cal I}_B$.
\item[B4]
We define the graph on ${\cal I}_B$ in the following way.
This graph cannot be divided.
$i_B\in {\cal I}_B$ is connected to $i_B' \in {\cal I}_B$ 
when the following two conditions holds.
\begin{description}
\item[B4-1]
The relation $\langle b_{i_B},b_{i_B'}\rangle \neq 0$ holds.
\item[B4-2]
The relation ${\cal I}_{A, i_B}\cap  {\cal I}_{A, i_B'} \neq \emptyset$ holds.

\end{description}
\end{description}

In the two qubit case, if 
The set $\{v_i\}_{i\in {\cal I}}$ of vectors 
contains the following $4$ vectors, then the 
conditions A1 and A2 hold;
\begin{align}
a_0 \otimes b_0,~
a_1 \otimes b_0,~
a_0 \otimes b_1,~
a_2 \otimes b_1,\label{XO1}
\end{align}
where $a_0\neq a_1,a_2$,
$ \langle b_0,b_1\rangle \neq 0$.

\begin{example}
CHSH inequality satisfies Conditions A1 and A2,
which can be checked by choosing the vectors in \eqref{XO1} as follows:
\begin{align}
&a_0= |A_{0,0}\rangle ,~
a_1=a_2= |A_{0,1}\rangle ,  \\
&b_0= |B_{0,0}\rangle ,~
b_1= |B_{1,0}\rangle .
\end{align}
\end{example}
\begin{example}
Chained Bell inequalities satisfies Conditions A1 and A2, 
which can be checked by choosing the vectors in \eqref{XO1} as follows:
\begin{align}
&a_0= |A_{1}=1\rangle ,~
a_1= |A_{3}=1\rangle ,~
a_2= |A_{2N-1}=-1\rangle , \\
&b_0= | B_2=1 \rangle ,~
b_1= |B_{2N}=-1\rangle .
\end{align}
In this example and the next example, $|A_{1}=1\rangle $ expresses the eigenvector of $A_1$ with eigenvalue $1$.
This notation is applied to other observables.
\end{example}
\begin{example}
Abner Shimony Self-Testing satisfies Conditions A1 and A2, 
which can be checked by choosing the vectors in \eqref{XO1} as follows:
\begin{align}
&a_0= |A_{2}=0\rangle ,~
a_1=a_2= |A_{3}=0\rangle ,  \\
&b_0= | B_0=0 \rangle ,~
b_1= |B_{1}=0\rangle .
\end{align}
\end{example}

\begin{theorem}\Label{NMP}
Assume that 
the optimal maximizer given in \eqref{NA1} satisfies conditions A1 and A2
and
the vectors $(\Pi_i |\psi'\rangle)_{i \in {\cal I}}$ realize the optimal solution in the SDP \eqref{thetaSDP}.
In addition, the ranks of the projections $\Pi_{i_A}^A$ and $\Pi_{i_B}^B$ 
are assumed to be one.
Then, there exist isometries 
$V_A: {\cal H}_A \to {\cal H}_A'$
and $V_B: {\cal H}_B\to {\cal H}_B'$ such that
\begin{align}
V_A \otimes V_B|\psi \rangle&= |\psi'\rangle , \\
V_A \otimes V_B|v_i \rangle&= |v_i'\rangle ,
\end{align}
for $i \in {\cal I}$.
\hfill $\square$\end{theorem}

\begin{proof}
Since the vectors $\Pi_i |\psi'\rangle$ realize the optimal solution in the SDP \eqref{thetaSDP},
there exists a isometry $V$ from ${\cal H}_A\otimes {\cal H}_B$
to ${\cal H}_A'\otimes {\cal H}_B'$ such that
\begin{align}
V \Pi_i|\psi\rangle=\Pi_i|\psi'\rangle \Label{APY}
\end{align}
for $i \in {\cal I}$.
We denote 
$ \Pi_i|\psi'\rangle=\eta_i |a_{i_A}' \otimes b_{i_B}'\rangle$.

We fix an arbitrary element $i_B \in {\cal I}_B$.
For $i_A ,i_A' \in {\cal I}_{A, i_B}$,
Condition A1 implies
\begin{align}
\langle  a_{i_A}, a_{i_A'} \rangle
=\langle  a_{i_A}\otimes b_{i_B}, a_{i_A'}\otimes b_{i_B} \rangle
=\langle  a_{i_A}'\otimes b_{i_B}', a_{i_A'}'\otimes b_{i_B}' \rangle
=\langle  a_{i_A}', a_{i_A'}' \rangle.
\end{align}
Hence, there exists an isometry $V_{A,i_B}: 
{\cal H}_A\to {\cal H}_A'$
such that
\begin{eqnarray}
V_{A,i_B} |a_{i_A}\rangle=|a_{i_A}'\rangle
\end{eqnarray}
for $i_A \in {\cal I}_{A, i_B}$.

We choose two connected elements $i_B,i_B' \in {\cal I}_B$.
For $i_A \in {\cal I}_{A, i_B}\cap  {\cal I}_{A, i_B'}$,
\eqref{APY} implies
\begin{align}
\langle  b_{i_B}, b_{i_B'} \rangle
=\langle  a_{i_A}\otimes b_{i_B}, a_{i_A}\otimes b_{i_B'} \rangle
=\langle  a_{i_A}'\otimes b_{i_B}', a_{i_A}'\otimes b_{i_B'}' \rangle
=\langle  b_{i_B}', b_{i_B'}' \rangle.\Label{G1}
\end{align}
Hence, 
for $i_A \in {\cal I}_{A, i_B}$ and $i_A' \in {\cal I}_{A, i_B'}$,
\eqref{APY} implies
\begin{align}
 \langle  a_{i_A}, a_{i_A'} \rangle \langle  b_{i_B}, b_{i_B'}\rangle 
=&\langle  a_{i_A}\otimes b_{i_B}, a_{i_A'}\otimes b_{i_B'} \rangle
=\langle  a_{i_A}'\otimes b_{i_B}', a_{i_A'}'\otimes b_{i_B'}' \rangle\nonumber \\
=&\langle  a_{i_A}', a_{i_A'}' \rangle
\langle  b_{i_B}', b_{i_B'}' \rangle.\Label{G2}
\end{align}
Since Condition B4-1 guarantees $\langle  b_{i_B}, b_{i_B'} \rangle \neq 0$, 
the combination of \eqref{G1} and \eqref{G2} implies that
\begin{align}
\langle  a_{i_A}, a_{i_A'} \rangle
=\langle  a_{i_A}', a_{i_A'}' \rangle.
\end{align}
Hence, we find that  $V_{A,i_B}=V_{A,i_B'}$.
Since the graph defined in B4 is not divided,
all isometries $V_{A,i_B}$ are the same.
We denote it by $V_A$.

We choose arbitrary two elements $i_B,i_B' \in {\cal I}_B$.
We choose elements $i_A \in {\cal I}_{A, i_B}$ and $i_A' \in {\cal I}_{A, i_B'}$
such that 
\begin{align}
\langle  a_{i_A}, a_{i_A'} \rangle \neq 0.\Label{G3}
\end{align}
Condition A1 implies
\begin{align}
 &\langle  a_{i_A}, a_{i_A'} \rangle \langle  b_{i_B}, b_{i_B'}\rangle
=\langle  a_{i_A}\otimes b_{i_B}, a_{i_A'}\otimes b_{i_B'} \rangle
=\langle  a_{i_A}'\otimes b_{i_B}', a_{i_A'}'\otimes b_{i_B'}' \rangle \nonumber \\
=&\langle  a_{i_A}', a_{i_A'}' \rangle
\langle  b_{i_B}', b_{i_B'}' \rangle
=\langle V_A a_{i_A},V_A a_{i_A'} \rangle \langle  b_{i_B}', b_{i_B'}' \rangle
=\langle a_{i_A}, a_{i_A'} \rangle \langle  b_{i_B}', b_{i_B'}' \rangle.
\Label{G4}
\end{align}
the combination of \eqref{G3} and \eqref{G4} implies that
\begin{align}
\langle  b_{i_B}, b_{i_B'} \rangle
=\langle  b_{i_B}', b_{i_B'}' \rangle.
\end{align}
Hence, there exists an isometry 
 $V_B: {\cal H}_B\to {\cal H}_B'$ such that
\begin{eqnarray}
V_B |b_{i_B}\rangle=|b_{i_B}'\rangle
\end{eqnarray}
for $i_B \in {\cal I}_B$.

Since $\{a_{i_A}\otimes b_{i_B}\}_{i_A \in {\cal I}_{A,i_B} ,i_B \in {\cal I}_B }$
spans ${\cal H}_A\otimes {\cal H}_B$,
we have
$V=V_A \otimes V_B$.
\end{proof}

\subsubsection{General case}
We consider the general case.
In addition to A1 and A2, we assume the following condition.
\begin{description}
\item[A3] Ideal systems ${\cal H}_A$ and ${\cal H}_B$ are two-dimensional.
\item[A4] Each system has only two measurements.
That is, the set $\bar{\cal I}_A$ ($\bar{\cal I}_B$) of all indexes of the space ${\cal H}_A$ (${\cal H}_B$) is composed 
4 elements.
For any element $i_A \in \bar{\cal I}_A$ ($i_B \in \bar{\cal I}_B$), there exists 
an element $i_A' \in \bar{\cal I}_A$ ($i_B' \in \bar{\cal I}_B$) such that 
$\langle a_{i_A}| a_{i_A'}\rangle=0$
($\langle b_{i_B}| b_{i_B'}\rangle=0$).
\end{description}

When A3 and A4 hold,
$\bar{\cal I}_A$ ($\bar{\cal I}_B$) is written as
${\cal B}_{A,0}\cup {\cal B}_{A,1}$ (${\cal B}_{B,0}\cup {\cal B}_{B,1}$), where
${\cal B}_{A,j}=\{ (0,j),(1,j)\}$ (${\cal B}_{B,j}=\{ (0,j),(1,j)\}$) 
and $\langle a_{(0,j)}| a_{(1,j)}\rangle=0$ ($\langle b_{(0,j)}| b_{(1,j)}\rangle=0$) 
for $j=0,1$.

While CHSH inequality, Chained Bell inequalities, and Abner Shimony Self-Testing
satisfy Conditions A1 and A2,  only  
CHSH inequality satisfies Conditions A3 and A4. 

We also consider the following condition for $\Pi_i= \Pi_{i_A}^A\otimes \Pi_{i_B}^B$.
\begin{description}
\item[C1]
When $i_A,i_A' \in \bar{\cal I}_A$ ($i_B,i_B' \in \bar{\cal I}_B$) satisfy
$\langle a_{i_A}| a_{i_A'}\rangle=0$ ($\langle b_{i_B}| b_{i_B'}\rangle=0$), 
we have $\Pi_{i_A}^A+\Pi_{i_A'}^A=I$ ($\Pi_{i_B}^B+\Pi_{i_B'}^B=I$). 
\end{description}

Let ${\cal H}_{i_A}^A$ and ${\cal H}_{i_B}^B$ be the image of the projections
$\Pi_{i_A}^A$ and $\Pi_{i_B}^B$.

\begin{theorem}\Label{APS}
Assume that 
the optimal maximizer given in \eqref{NA1} satisfies conditions A1, A2, A3, and A4,
the vectors $(\Pi_i |\psi'\rangle)_{i \in {\cal I}}$ realize the optimal solution in the SDP \eqref{thetaSDP},
and condition C1 holds.
Then, there exist isometries 
$V_A$ from $ {\cal H}_A\otimes  {\cal K}_A$ to $ {\cal H}_A'$
and 
$V_B$ from $ {\cal H}_B\otimes  {\cal K}_B$ to $ {\cal H}_B'$
such that
\begin{align}
V_A \otimes V_B |\psi\rangle \otimes |junk\rangle &=
|\psi'\rangle , \Label{ML1}\\
V_A \otimes V_B |v_i\rangle \otimes |junk\rangle &=|v_i'\rangle , \Label{ML2}
\end{align}
for $i \in {\cal I}$,
where
$|junk\rangle$ is a state on ${\cal K}_A\otimes  {\cal K}_B$.
\hfill $\square$\end{theorem}

\begin{lemma}\Label{BA7}
Assume that the vectors $(\Pi_i |\psi'\rangle)_{i \in {\cal I}}$ realize the optimal solution in the SDP \eqref{thetaSDP}.
Assume that a projection $\Pi$ is commutative with 
$\Pi_i$ for any $i \in {\cal I}$.
Also assume that $ \Pi \psi'\neq 0$.
Let $\psi'(\Pi)$ be the normalized vector of $ \Pi \psi'$.
Then, 
the vectors $\Pi_i |\psi'(\Pi)\rangle$ realize the optimal solution in the SDP \eqref{thetaSDP}.
\hfill $\square$\end{lemma}

\begin{proofof}{Lemma \ref{BA7}}
\begin{align}
&\sum_{i}\langle \psi' |\Pi_i|\psi'\rangle
=
\sum_{i}\langle \Pi \psi' |\Pi_i|\Pi \psi'\rangle
+\sum_{i}\langle (I-\Pi)\psi' |\Pi_i|(I-\Pi)\psi'\rangle \\
=&
\|\Pi \psi'\|^2\sum_{i}\langle \psi'(\Pi) |\Pi_i|\psi'(\Pi) \rangle
+\|(I-\Pi) \psi'\|^2\sum_{i}\langle \psi' (I-\Pi)|\Pi_i|\psi'(I-\Pi)\rangle .\Label{APO}
\end{align}

Since the vectors $(\Pi_i |\psi'(\Pi)\rangle)_{i \in {\cal I}}$
and 
the vectors $(\Pi_i |\psi'(I-\Pi)\rangle)_{i \in {\cal I}}$
satisfy the condition of the SDP \eqref{thetaSDP},
\eqref{APO} shows that 
either 
the vectors $(\Pi_i |\psi'(\Pi)\rangle)_{i \in {\cal I}}$
or 
the vectors $(\Pi_i |\psi'(I-\Pi)\rangle)_{i \in {\cal I}}$
realizes the optimal solution in the SDP \eqref{thetaSDP}.
Hence, the remaining one of 
the vectors $(\Pi_i |\psi'(\Pi)\rangle)_{i \in {\cal I}}$
and
the vectors $(\Pi_i |\psi'(I-\Pi)\rangle)_{i \in {\cal I}}$
also
realizes the optimal solution in the SDP \eqref{thetaSDP}.
\end{proofof}

Considering the contraposition of Lemma \ref{BA7}, we have the following lemma.

\begin{lemma}\Label{BA8}
Assume that the vectors $\Pi_i |\psi'\rangle$ realize the optimal solution in the SDP \eqref{thetaSDP}.
Assume that a projection $\Pi$ is commutative with 
$\Pi_i$ for any $i \in {\cal I}$.
Also, there exists an element $j \in {\cal I}$ such that 
$\Pi \Pi_j=0$.
Then, $ \Pi \psi'=0$.
\hfill $\square$\end{lemma}

\begin{proofof}{Theorem \ref{APS}}

\noindent{\bf Step 1:}\quad
Let $P_{A,(0,0),(0,1) }$ be the projection 
to the eigenspace of 
$\Pi_{(0,0)}^A \Pi_{(0,1)}^A \Pi_{(0,0)}^A $ with one eigenvalue.
Let $P_{A,(0,0),(1,1) }$ be the projection 
to the eigenspace of 
$\Pi_{(0,0)}^A \Pi_{(0,1)}^A \Pi_{(0,0)}^A $ with zero eigenvalue.
Let $\{ e_{j_A}^A\}$ be orthogonal basis corresponding to orthogonal eigenvectors of 
$\Pi_{(0,0)}^A \Pi_{(0,1)}^A \Pi_{(0,0)}^A $ with other eigenvalues.
We define $f_{j_A}^A$ as the normalized vector of $ \Pi_{(0,1)}^A e_{j_A}^A$.
For $j_A\neq j_A'$, $f_{j_A}^A$ is orthogonal to $f_{j_A'}^A$
due to the choice of $\{ e_{j_A}^A\}$.
We define $g_{j_A}^A$ as the normalized vector of $
f_{j_A}^A- \langle e_{j_A}^A,f_{j_A}^A \rangle e_{j_A}^A$.
$g_{j_A}^A$ belongs to ${\cal H}^A_{(1,0)}$.
For $j_A\neq j_A'$, $g_{j_A}^A$ is orthogonal to $g_{j_A'}^A$
because $e_{j_A}^A$ and $f_{j_A}^A$ are orthogonal to $e_{j_A'}^A$ and $f_{j_A'}^A$,
respectively.
We define the projection 
$\bar\Pi_{j_A}^A:= |e_{j_A}^A\rangle \langle e_{j_A}^A|+|g_{j_A}^A\rangle \langle g_{j_A}^A|$.
For $j_A\neq j_A'$, we have
$\bar \Pi_{j_A}^A \bar\Pi_{j_A'}^A=0$.
$\bar \Pi_{j_A}^A$ is commutative with $\Pi_{(0,0)}^A$, $\Pi_{(1,0)}^A$, $\Pi_{(0,1)}^A$, and $\Pi_{(1,1)}^A$.
We define 
$\Pi^A:= \sum_{j_A} \bar\Pi_{j_A}^A$.
Also, the projections $\Pi^A$, $P_{A,(0,0),(0,1) }$, and $P_{A,(0,0),(1,1) }$
is commutative with $\Pi_{(0,0)}^A$, $\Pi_{(1,0)}^A$, $\Pi_{(0,1)}^A$, and $\Pi_{(1,1)}^A$.
Since 
$(I- \Pi^A-P_{A,(0,0),(0,1) }-P_{A,(0,0),(1,1) })\Pi_{(0,0)}^A =0$
$P_{A,(0,0),(0,1) }\Pi_{(1,1)}^A =0$, and $P_{A,(0,0),(1,1) }\Pi_{(0,1)}^A =0$,
Lemma \ref{BA8} implies that
$(I- \Pi^A-P_{A,(0,0),(0,1) }-P_{A,(0,0),(1,1) })\psi'=0$,
$P_{A,(0,0),(0,1) }\psi'=0$, and
$P_{A,(0,0),(1,1) }\psi'=0$.
Hence, we have $ \Pi^A \psi'=\psi'$.

In the same way, we define the projections $ \bar\Pi_{j_B}^B$ and $\bar \Pi^B$.
We define the projection $\bar\Pi_{(j_A,j_B)}:= \bar\Pi_{j_A}^A \bar\Pi_{j_B}^B$.
$\bar\Pi_{(j_A,j_B)}$ is commutative with $\Pi_i$ for $i \in {\cal I}$.
When $\bar\Pi_{(j_A,j_B)}\psi' \neq 0$, we define
$ \psi_{(j_A,j_B)}:= \alpha_{(j_A,j_B)}\bar\Pi_{(j_A,j_B)}\psi' $,
where $\alpha_{(j_A,j_B)}:=\|\bar\Pi_{(j_A,j_B)}\psi' \|^{-1}$.

\noindent{\bf Step 2:}\quad
Due to Lemma \ref{BA7}, 
the vectors $\Pi_i\psi_{(j_A,j_B)}= \Pi_i\bar\Pi_{(j_A,j_B)}\psi_{(j_A,j_B)}$ 
realize the optimal solution in the SDP \eqref{thetaSDP}.
Also, $\Pi_i\bar\Pi_{(j_A,j_B)}$ is rank-one.
Hence, we can apply Theorem \ref{NMP} to
the vectors $\Pi_i\bar\Pi_{(j_A,j_B)}\psi_{(j_A,j_B)}$.
Thus, there exists isometries 
$V_{A,(j_A,j_B)}: {\cal H}_A \to \im \Pi_{j_A}^A$
and $V_{B,(j_A,j_B)}: {\cal H}_B\to \im \Pi_{j_B}^B$ 
such that
\begin{align}
V_{A,(j_A,j_B)} \otimes V_{B,(j_A,j_B)}\psi &=
\psi_{(j_A,j_B)} . \\
\eta_i (V_{A,(j_A,j_B)} a_{i_A})\otimes(V_{B,(j_A,j_B)} b_{i_B})
&=V_{A,(j_A,j_B)} \otimes V_{B,(j_A,j_B)} (\eta_i a_{i_A}\otimes b_{i_B}) \\
&= \Pi_i \bar\Pi_{(j_A,j_B)}\psi_{(j_A,j_B)}\\
&= \Pi_{i_A}^A\bar\Pi_{j_A}^A\otimes \Pi_{i_B}^B\bar\Pi_{j_B}^B\psi_{(j_A,j_B)}.
\end{align}
As shown in Step 3, for $j_B\neq j_B'$, we have
$V_{A,(j_A,j_B)}=\beta_{j_A,j_B,j_B'} V_{A,(j_A,j_B')}$
with a constant $\beta_{j_A,j_B,j_B'}$
when $\Pi_{(j_A,j_B)}\psi' \neq 0$ and $\Pi_{(j_A,j_B')}\psi' \neq 0$.
That is, 
\begin{align}
V_{A,(j_A,j_B)} \otimes V_{B,(j_A,j_B')}\psi &=
\beta_{j_A,j_B,j_B'}\psi_{(j_A,j_B')} . 
\end{align}

Then, for $j_A$, we choose an element $j_B$ such that 
$\Pi_{(j_A,j_B)}\psi' \neq 0$.
Then, we define $V_{A,j_A}:=V_{A,(j_A,j_B)}$.
Thus, for elements $j_A'$ and $j_B'$, there exists an constant 
$\beta_{j_A',j_B'}$ such that
\begin{align}
V_{A,j_A'} \otimes V_{B,j_B'}\psi =& \beta_{j_A',j_B'}\psi_{(j_A',j_B')} 
= \beta_{j_A',j_B'} \alpha_{(j_A',j_B')}\bar\Pi_{(j_A',j_B')}\psi' \\
=& \beta_{j_A',j_B'} \alpha_{(j_A',j_B')}\bar\Pi_{j_A'}^A \bar\Pi_{j_B'}^B \psi' .
\end{align}
Hence, we have
\begin{align}
\beta_{j_A',j_B'}^{-1} \alpha_{(j_A',j_B')}^{-1} V_{A,j_A'} \otimes V_{B,j_B'}\psi =
 \bar\Pi_{j_A'}^A \bar\Pi_{j_B'}^B \psi' .
\end{align}
We define the spaces ${\cal K}_A$ and ${\cal K}_B$
spanned by $\{|j_A\rangle\} $ and $\{|j_B\rangle\} $, respectively.
We define the junk state on ${\cal K}_A\otimes {\cal K}_B$ as
\begin{align}
|junk\rangle:= \sum_{ j_A,j_B} \beta_{j_A,j_B}^{-1} \alpha_{(j_A,j_B)}^{-1}
|j_A,j_B\rangle.
\end{align}
We define the isometries $V_A: {\cal H}_A\otimes {\cal K}_A\to {\cal H}_A' $ 
and $V_B: {\cal H}_B\otimes {\cal K}_B\to {\cal H}_B'$ as
\begin{align}
V_A:= \sum_{j_A} V_{A,j_A} \langle j_A|,\quad
V_B:= \sum_{j_B} V_{B,j_B} \langle j_B|.
\end{align}
The isometries $V_A$ and $V_B$ satisfy conditions \eqref{ML1} and \eqref{ML2}.

\noindent{\bf Step 3:}\quad
We show the following fact;
For $j_B\neq j_B'$, we have
$V_{A,(j_A,j_B)}=\beta_{j_A,j_B,j_B'} V_{A,(j_A,j_B')}$
with a constant $\beta_{j_A,j_B,j_B'}$
when $\Pi_{(j_A,j_B)}\psi' \neq 0$ and $\Pi_{(j_A,j_B')}\psi' \neq 0$.

We define $a_{i_A,j_A,j_B}:=V_{A,(j_A,j_B)} a_{i_A}$.
Then, we have
\begin{align}
\Pi_i \bar\Pi_{(j_A,j_B')}\psi_{(j_A,j_B')}
= \Pi_{i_A}^A\bar\Pi_{j_A}^A\otimes \Pi_{i_B}^B\bar\Pi_{j_B'}^B\psi_{(j_A,j_B')}.
\end{align}
The above vector is a constant times of 
$\eta_i a_{i_A,j_A,j_B}\otimes b_{i_B,j_A,j_B'}$.
Also, 
the vectors $(\eta_i a_{i_A,j_A,j_B}\otimes b_{i_B,j_A,j_B'})_i$
and the vectors $(\Pi_i \bar\Pi_{(j_A,j_B')}\psi_{(j_A,j_B')})_i$
are the unique optimal solution in the SDP \eqref{thetaSDP}.
Hence, there exists a constant $\beta_{j_A,j_B,j_B'}$ such that 
$\eta_i a_{i_A,j_A,j_B}\otimes b_{i_B,j_A,j_B'}=\beta_{j_A,j_B,j_B'} 
\Pi_i \bar\Pi_{(j_A,j_B')}\psi_{(j_A,j_B')}$,
which is the desired statement.
\end{proofof}

\subsection{Tripartite case}
We assume that the unique optimal maximizer 
$X^*=(X_{ij})$ is given by $\eta_i \eta_j\langle v_j, v_i\rangle$ with the following;
For $i=(i_A,i_B,i_C)\in {\cal I}$, 
\begin{align}
 v_{i} = a_{i_A} \otimes b_{i_B}\otimes c_{i_C}, 
\label{MA8}
\end{align}
where $a_{i_A} \in {\cal H}_A=\mathbb{C}^{d_A}$, 
$b_{i_B} \in {\cal H}_B=\mathbb{C}^{d_B}$,
$c_{i_C} \in {\cal H}_C=\mathbb{C}^{d_C}$.
Also, for simplicity, $a_{i_A} $, $ b_{i_B}$, and $ c_{i_C}$ 
are assumed to be normalized and $\eta_i>0$.

Now, we consider a state $|\psi'\rangle $ on 
${\cal H}_A'\otimes {\cal H}_B'\otimes {\cal H}_C'$,
and 
projections $\Pi_{i_A}^A$, $\Pi_{i_B}^B$, $\Pi_{i_C}^C$ 
on ${\cal H}_A'$, ${\cal H}_B'$, and ${\cal H}_C'$.
Then, we define the projection 
$\Pi_i:=\Pi_{i_A}^A\otimes \Pi_{i_B}^B\otimes \Pi_{i_C}^B$.

In the following, 
we discuss how 
the state $ |\psi'\rangle$ is locally converted to 
$ |\psi\rangle$ when
the vectors $\Pi_i |\psi'\rangle$ realize the optimal solution in the SDP \eqref{thetaSDP}.
We define $|v_i'\rangle:= \eta_i^{-1} \Pi_{i}|\psi'\rangle$.

\subsubsection{Rank-one case}
We consider the case that the ranks of the projections 
$\Pi_{i_A}^A$, $\Pi_{i_B}^B$ and $\Pi_{i_C}^C$ 
are one.
We introduce the following conditions.

\begin{definition}\Label{D1}
Three distinct elements $i,j,k\in {\cal I}$ are called {\it linked}
when the following two conditions holds.
\begin{description}
\item[C1]
The relations 
$\langle v_i,v_{k}\rangle \neq 0$,
$\langle v_i,v_{j}\rangle \neq 0$, and
$\langle v_j,v_{k}\rangle \neq 0$
hold.
\item[C2]
$v_i,v_{j}$ shares a $t_{i,j}-$th common element for $t_{i,j} \in \{A,B,C\}$.
Other components of $v_i,v_{j}$ are different.
That is, when $t_{i,j}=A$, $i_A=j_A$,$i_B\neq j_B$,and $i_C\neq j_C$.
$v_i$ and $v_{k}$ share a $t_{i,k}-$th common element for 
$t_{i,k} \in \{A,B,C\}\setminus \{t_{i,j}\}$.
$v_j,v_{k}$ shares a $t_{j,k}-$th common element for 
$t_{j,k} \in \{A,B,C\}\setminus \{t_{i,j},t_{i,k}\}$.
In this case, there exist elements 
$x_A,x_A',x_B,x_B',x_C,x_C'$ such that
$i,j,k \in \{x_A,x_A'\}\times \{x_B,x_B'\} \times \{x_C,x_C'\}$.
\end{description}
In addition, 
two distinct elements $x_A,x_A'$ for index of a vectors of $\mathbb{C}^{d_A}$
are called {\it connected}
when there exist three linked elements $i,j,k\in {\cal I}$ 
such that the first components of $i,j,k\in {\cal I}$ are $x_A,x_A'$.
\hfill $\square$\end{definition}

For $i_B,i_C$, we use notation
\begin{align}
\psi_{(i_B,i_C)}:= b_{i_B} \otimes c_{i_C}.
\end{align}
Then, we introduce the following conditions for the optimal maximizer given in \eqref{MA8}.

\begin{description}
\item[A5]
The vectors $\{v_i\}_{i\in {\cal I}_0}$ 
span the vector space ${\cal H}_A\otimes {\cal H}_B\otimes {\cal H}_C$.

\item[A6]
There exist
a subset ${\cal I}_A$ of indexes of the space 
${\cal H}_A$ with $|{\cal I}_A|=d_A$
and 
$d_A$ sets ${\cal I}_{BC,i_A}$ for $i_A \in {\cal I}_A$ of indexes of the space 
${\cal H}_B\otimes {\cal H}_C$ 
to satisfy the following conditions.
The set $\{a_{i_A}\}_{i_A\in {\cal I}_A}$
spans the space ${\cal H}_A$.
The set $\{\psi_{i_{BC}}\}_{i_{BC}\in {\cal I}_{BC,i_A}}$
spans the space ${\cal H}_B\otimes {\cal H}_C$ and 
${\cal I}_0= \cup_{i_A \in {\cal I}_A} (\{i_A\} \times {\cal I}_{BC,i_A})$.
We consider the graph $G_A$ with the set ${\cal I}_A$ of vertecies 
such that 
the edges are given as the the pair of 
all connected elements in ${\cal I}_A$ in the sense of the end of Definition \ref{D1}.
The graph $G_A$ is not divided into two disconnected parts.

\item[A7]
The vectors
$\{b_{i_B}\otimes c_{i_C}\}_{
(i_B,i_C)\in \cup_{i_A \in {\cal I}_A}{\cal I}_{BC,i_A}}$ 
satisfy condition A2 by substituting $c_{i_C}$ into $a_{i_A}$.
That is, there exist a subset ${\cal I}_B$ of the second indexes and
subsets ${\cal I}_{C,i_B}$ of the third indexes such that
they satisfy conditions B1, B2, B3, and B4.
We denote the graph defined in this condition by $G_B$

\end{description}

\begin{example}
We can check that Mermin Self-testing satisfies Conditions A5, A6, and A7 as follows.
In this example, $a_O,b_O,c_O$ means $|O\rangle$.
This notation is applied to $Z,P,M$.

We choose the subset ${\cal I}_A:=\{ O,P \}$.
Then, we have
\begin{align}
{\cal I}_{BC,O} &=\{ (O,O),(Z,Z),(M,P),(P,M)\},\label{MI1}\\
{\cal I}_{BC,P} &=\{ (Z,P),(P,Z),(O,M),(M,O)\}.\label{MI2}
\end{align}
Two elements $O,P \in {\cal I}_A$ are connected in the sense of the end of Definition \ref{D1} by choosing 
$\{i,j,k\}=\{ (P,Z,P),(O,Z,Z), (O,M,P) \}$.
Based on \eqref{MI1} and \eqref{MI2}, 
we choose the subsets ${\cal I}_B$, ${\cal I}_{C,Z}$, and ${\cal I}_{C,P}$ as
\begin{align}
{\cal I}_B:= \{Z,P\},~ {\cal I}_{C,Z}:=\{Z,P\},~
{\cal I}_{C,P}:=\{Z,M\}.
\end{align}
The subsets ${\cal I}_B$, ${\cal I}_{C,Z}$, and ${\cal I}_{C,P}$
satisfy conditions B1, B2, B3, and B4.
\end{example}

\begin{lemma}\Label{L4}
Assume that 
$i,j,k \in {\cal I}_0$ are connected by one edge, i.e., satisfy 
conditions C1 and C2.
We choose $x_A,x_A',x_B,x_B',x_C,x_C'$ in the way as Condition C2.
We consider three normalized vectors
$v_i',v_j',v_k'$, where
\begin{align}
v_l'=a_{l_A}\otimes b_{l_B}\otimes c_{l_C}
\end{align}
for $l=i,j,k$.
We assume that
$\langle v_l,v_{l}\rangle =\langle v_l',v_{l}'\rangle $
for $l,l'=i,j,k$.
Then, we have 
\begin{align}
\langle a_{x_A},a_{x_A'}\rangle &=\langle a_{x_A}',a_{x_A'}'\rangle\Label{E1} \\
\langle b_{x_B},b_{x_B'}\rangle &=\langle b_{x_B}',b_{x_B'}'\rangle\Label{E2} \\
\langle c_{x_C},c_{x_C'}\rangle &=\langle c_{x_C}',c_{x_C'}'\rangle \Label{E3}
\end{align}
or
\begin{align}
\langle a_{x_A},a_{x_A'}\rangle &=-\langle a_{x_A}',a_{x_A'}'\rangle \Label{E4}\\
\langle b_{x_B},b_{x_B'}\rangle &=-\langle b_{x_B}',b_{x_B'}'\rangle \Label{E5}\\
\langle c_{x_C},c_{x_C'}\rangle &=-\langle c_{x_C}',c_{x_C'}'\rangle .\Label{E6}
\end{align}\hfill $\square$
\end{lemma}

\begin{proof}
For simplicity, without loss of generality, we assume that
\begin{align}
i=(x_A',x_B,x_C),~
j=(x_A,x_B',x_C),~
k=(x_A,x_B,x_C').
\end{align}
Since
\begin{align}
\langle v_i,v_j \rangle=\langle v_i',v_j' \rangle,~
\langle v_i,v_k \rangle=\langle v_i',v_k' \rangle,~
\langle v_k,v_j \rangle=\langle v_k',v_j' \rangle,
\end{align}
we have
\begin{align}
\langle a_{x_A},a_{x_A'}\rangle \langle b_{x_B},b_{x_B'}\rangle
&=\langle a_{x_A}',a_{x_A'}'\rangle \langle b_{x_B}',b_{x_B'}'\rangle, \\
\langle a_{x_A},a_{x_A'}\rangle \langle c_{x_C},c_{x_C'}\rangle
&=\langle a_{x_A}',a_{x_A'}'\rangle \langle c_{x_C}',c_{x_C'}'\rangle, \\
\langle b_{x_B},b_{x_B'}\rangle \langle c_{x_C},c_{x_C'}\rangle 
&=\langle b_{x_B}',b_{x_B'}'\rangle \langle c_{x_C}',c_{x_C'}'\rangle .
\end{align}
Hence,
\begin{align}
&\langle a_{x_A},a_{x_A'}\rangle^2 \\
=&
(\langle a_{x_A},a_{x_A'}\rangle \langle b_{x_B},b_{x_B'}\rangle)
(\langle a_{x_A},a_{x_A'}\rangle \langle c_{x_C},c_{x_C'}\rangle)
(\langle b_{x_B},b_{x_B'}\rangle \langle c_{x_C},c_{x_C'}\rangle)^{-1}\\
=&
(\langle a_{x_A}',a_{x_A'}'\rangle \langle b_{x_B}',b_{x_B'}'\rangle)
(\langle a_{x_A}',a_{x_A'}'\rangle \langle c_{x_C}',c_{x_C'}'\rangle)
(\langle b_{x_B}',b_{x_B'}'\rangle \langle c_{x_C}',c_{x_C'}'\rangle)^{-1} \\
=&
\langle a_{x_A}',a_{x_A'}'\rangle^2 ,
\end{align}
which implies \eqref{E1} or \eqref{E4}.
When \eqref{E1}, we have \eqref{E2} and \eqref{E3}.
When \eqref{E3}, we have \eqref{E4} and \eqref{E5}.
\end{proof}

\begin{theorem}\Label{TH9}
Assume that the optimal maximizer given in \eqref{MA8} satisfies
conditions A5, A6, and A7, and
the vectors $(\Pi_i |\psi'\rangle)_{i \in {\cal I}}$ realize the optimal solution in the SDP \eqref{thetaSDP}.
In addition, the ranks of the projections 
$\Pi_{i_A}^A$, $\Pi_{i_B}^B$, and $\Pi_{i_C}^C$ 
are assumed to be one.

Then, there exist isometries 
$V_A: {\cal H}_A \to {\cal H}_A'$,
$V_B: {\cal H}_B\to {\cal H}_B'$, 
and $V_C: {\cal H}_C\to {\cal H}_C'$ 
such that
\begin{align}
V_A \otimes V_B\otimes V_C|\psi \rangle=&
|\psi'\rangle , \\
V_A \otimes V_B\otimes V_C|v_i \rangle=&
|v_i'\rangle ,
\end{align}
for $i \in {\cal I}$.
\end{theorem}
\begin{proof}

\noindent{\bf Step 1:}\quad
We fix an arbitrary element $i_A \in {\cal I}_A$.
For $i_{BC} ,i_{BC}' \in {\cal I}_{BC, i_A}$,
Condition A1 implies
\begin{align}
\langle  \psi_{i_{BC}}, \psi_{i_{BC}'} \rangle
=\langle  a_{i_A}\otimes \psi_{i_{BC}}, a_{i_A}\otimes \psi_{i_{BC}'} \rangle
=\langle  a_{i_A}'\otimes \psi_{i_{BC}}', a_{i_A}'\otimes \psi_{i_{BC}'}' \rangle
=\langle  \psi_{i_{BC}}', \psi_{i_{BC}'}' \rangle.
\end{align}
Hence, there exists an isometry $V_{BC,i_A}: 
{\cal H}_B \otimes {\cal H}_C \to {\cal H}_B' \otimes {\cal H}_C'$
such that
\begin{eqnarray}
V_{BC,i_A} \psi_{i_{BC}}=\psi_{i_{BC}}'
\end{eqnarray}
for $i_{BC} \in {\cal I}_{BC, i_A}$.

{\bf Step 2:}\quad
We choose a subgraph $G_{A,0} \subset G_A$ such that
the vertecies of $G_{A,0}$ is ${\cal I}_A$,
$G_{A,0}$ has no cycle, and
$G_{A,0}$ cannot be divided into two parts.

We fix the origin $i_{A,0} \in {\cal I}_A$.
For any element $ i_A \in {\cal I}_A$, we have the unique path 
to connect $i_{A,0}$ and $i_A$ 
by using $G_{A,0}$ because $G_{A,0}$ has no cycle.
We denote this path as
$i_{A,0}-i_{A,1}-\cdots -i_{A,n}=i_A$.
We define $\alpha(i_A)$ as
\begin{align}
\alpha(i_A):= \prod_{m=1}^{n}  \frac{\langle a_{i_{A,m-1}},a_{i_{A,m}}\rangle}
{\langle a_{i_{A,m-1}}',a_{i_{A,m}}'\rangle}.
\end{align}
Lemma \ref{L4} guarantees that $\alpha(i_A)$ takes value $1$ or $-1$.
Due to the above definition and the uniqueness of the above path, we find that
\begin{align}
\alpha(i_{A,l}):= \prod_{m=1}^{l}  \frac{\langle a_{i_{A,m-1}},a_{i_{A,m}}\rangle}
{\langle a_{i_{A,m-1}}',a_{i_{A,m}}'\rangle}.
\end{align}

For $i_{BC} \in {\cal I}_{BC,i_{A,l} }$
and $i_{BC}' \in {\cal I}_{BC,i_{A,l+1} }$,
we find that
\begin{align}
&\langle a_{i_{A,l}},a_{i_{A,l+1}}\rangle \langle \psi_{i_{BC}},\psi_{i_{BC}'}\rangle
=\langle a_{i_{A,l}}\otimes \psi_{i_{BC}},a_{i_{A,l+1}}\otimes \psi_{i_{BC}'}\rangle \\
=&
\langle a_{i_{A,l}}'\otimes \psi_{i_{BC}}',a_{i_{A,l+1}}'\otimes \psi_{i_{BC}'}'\rangle \\
=&
\langle a_{i_{A,l}}',a_{i_{A,l+1}}'\rangle \langle \psi_{i_{BC}}',\psi_{i_{BC}'}'\rangle \\
=&
\alpha(i_{A,l})\alpha(i_{A,l+1})
\langle a_{i_{A,l}},a_{i_{A,l+1}}\rangle \langle V_{BC,i_{A,l}}\psi_{i_{BC}},V_{BC,i_{A,l+1}}\psi_{i_{BC}'}\rangle \\
=&
\alpha(i_{A,l})\alpha(i_{A,l+1})
\langle a_{i_{A,l}},a_{i_{A,l+1}}\rangle \langle 
\psi_{i_{BC}},
V_{BC,i_{A,l}}^\dagger V_{BC,i_{A,l+1}}
\psi_{i_{BC}'}\rangle.
\end{align}
Since $\langle a_{i_{A,l}},a_{i_{A,l+1}}\rangle\neq 0$ and
the sets 
$\{\psi_{i_{BC}}\}_{i_{BC} \in {\cal I}_{BC,i_{A,l} }}$
and 
$\{\psi_{i_{BC}'}\}_{i_{BC}' \in {\cal I}_{BC,i_{A,l+1} }}$
span the space $\mathbb{C}^{d_Bd_C}$,
we find that
$\alpha(i_{A,l})\alpha(i_{A,l+1}) V_{BC,i_{A,l}}^\dagger V_{BC,i_{A,l+1}}$
is identity.
Then, we find that
\begin{align}
V_{BC}:=V_{BC,i_{A,0}}=\alpha(i_{A,l}) V_{BC,i_{A,l}}.
\end{align}
That is, we have 
\begin{align}
V_{BC}=\alpha(i_A) V_{BC,i_A}.\Label{BLA}
\end{align}
Also, we define the isometry $V_A: 
{\cal H}_A \to {\cal H}_A'$
such that
\begin{eqnarray}
V_A a_{i_A}= \alpha(i_A) a_{i_A}'
\end{eqnarray}
for $i_A \in {\cal I}_A$.

Therefore, for $(i_A,i_{BC}) \in \cup_{i_A \in {\cal I}_A }(\{i_A\} \times {\cal I}_{BC, i_A})$,
we have
\begin{align}
V a_{i_A}\otimes \psi_{i_{BC}}=a_{i_A}'\otimes \psi_{i_{BC}}'=
(V_A\otimes V_{BC})a_{i_A}\otimes \psi_{i_{BC}}.
\end{align}
Since the set $\{a_{i_A}\}_{i_A\in {\cal I}_A}$
spans the space $\mathbb{C}^{d_A}$,
we have
\begin{align}
V=V_A\otimes V_{BC}.\Label{NVR}
\end{align}

{\bf Step 3:}\quad
For $i_{B} \in {\cal I}_B$ and $i_{C} \in {\cal I}_{C, i_B}$,
we choose $i_A$ such that $ (i_B,i_C)\in {\cal I}_{BC,i_A}$.
Then, we define $\beta(i_B,i_C):= \alpha(i_A)$.
We fix an arbitrary element $i_{B} \in {\cal I}_B$.
For $i_{C} ,i_{C}' \in {\cal I}_{C, i_B}$,
Relation \eqref{BLA} implies
\begin{align}
\langle  c_{i_{C}}, c_{i_{C}'} \rangle
=&\langle  b_{i_B}\otimes c_{i_{C}}, b_{i_B}\otimes c_{i_{C}'} \rangle
=\beta(i_B,i_C)\beta(i_B,i_C')
\langle  b_{i_B}'\otimes c_{i_{C}}', b_{i_B}'\otimes c_{i_{C}'}' \rangle \\
=&\beta(i_B,i_C)\beta(i_B,i_C')\langle c_{i_{C}}', c_{i_{C}'}' \rangle.
\end{align}
Hence, there exists an isometry $V_{C,i_B}: 
{\cal H}_C \to {\cal H}_C'$
such that
\begin{eqnarray}
V_{C,i_B} c_{i_{C}}=\beta(i_B,i_C) c_{i_{C}}'
\end{eqnarray}
for $i_{C} \in {\cal I}_{C, i_B}$.

{\bf Step 4:}\quad
We choose a subgraph $G_{B,0} \subset G_B$ such that
the vertecies of $G_{B,0}$ is ${\cal I}_B$,
$G_{B,0}$ has no cycle, and
$G_{B,0}$ cannot be divided into two parts.

We fix the origin $i_{B,0} \in {\cal I}_B$.
For any element $ i_{B} \in {\cal I}_B$, we have the unique path 
to connect $i_{B,0}$ and $i_{B}$ 
by using $G_{B,0}$ because $G_{B,0}$ has no cycle.
We denote this path as
$i_{B,0}-i_{B,1}-\cdots -i_{B,n'}=i_B$.
We choose a non-zero element $i_{C,l} \in 
{\cal I}_{C,i_{B,l-1}}\cap {\cal I}_{C,i_{B,l}}$.
We choose $i_{A,l},i_{A,l}' $ such that $(i_{B,l-1}, i_{C,l})\in {\cal I}_{BC,i_{A,l}} $
and $(i_{B,l}, i_{C,l})\in {\cal I}_{BC,i_{A,l}'} $.
We define $\beta(i_B)$ as
\begin{align}
\gamma(i_B):= \prod_{l=1}^{n'} 
\beta(i_{B,l-1},i_{C,l})\beta(i_{B,l},i_{C,l}).
\end{align}
Then, we have
\begin{align}
&\langle b_{i_{B,l-1}},b_{i_{B,l}}\rangle 
=\langle b_{i_{B,l-1}}\otimes c_{i_{C,l}},b_{i_{B,l}}\otimes c_{i_{C,l}}\rangle \\
=&
\beta(i_{B,l-1},i_{C,l})\beta(i_{B,l},i_{C,l})
\langle b_{i_{B,l-1}}'\otimes c_{i_{C,l}}',b_{i_{B,l}}'\otimes c_{i_{C,l}}'\rangle \\
=&
\beta(i_{B,l-1},i_{C,l})\beta(i_{B,l},i_{C,l})
\langle b_{i_{B,l-1}}',b_{i_{B,l}}'\rangle \\
=&
\gamma(i_{B,l-1}) \gamma(i_{B,l})
\langle b_{i_{B,l-1}}',b_{i_{B,l}}'\rangle .
\end{align}

For $i_{C} \in {\cal I}_{C,i_{B,l} }$
and $i_{C}' \in {\cal I}_{C,i_{B,l+1} }$,
we find that
\begin{align}
&\langle b_{i_{B,l}},b_{i_{B,l+1}}\rangle \langle c_{i_{C}},c_{i_{C}'}\rangle
=\langle b_{i_{B,l}}\otimes c_{i_{C}},b_{i_{B,l+1}}\otimes c_{i_{C}'}\rangle \\
=&
\beta(i_{B,l},i_{C})\beta(i_{B,l+1},i_{C}')
\langle b_{i_{B,l}}'\otimes c_{i_{C}}',b_{i_{B,l+1}}'\otimes c_{i_{C}'}'\rangle \\
=&
\langle b_{i_{B,l}}',b_{i_{B,l+1}}'\rangle 
\beta(i_{B,l},i_{C})\beta(i_{B,l+1},i_{C}')
\langle c_{i_{C}}',c_{i_{C}'}'\rangle \\
=&
\gamma(i_{B,l}) \gamma(i_{B,l+1})
\langle b_{i_{B,l}},b_{i_{B,l+1}}\rangle 
\langle V_{C,i_{B,l}} c_{i_{C}},V_{C,i_{B,l+1}} c_{i_{C}'}\rangle .
\end{align}
Since $\langle b_{i_{B,l}},b_{i_{B,l+1}}\rangle\neq 0$ and
the sets 
$\{ c_{i_{C}}\}_{i_{C} \in {\cal I}_{C,i_{B,l} }}$
and 
$\{c_{i_{C}'}\}_{i_{C}' \in {\cal I}_{C,i_{B,l+1} }}$
span the space ${\cal H}_C$,
we find that
$\gamma(i_{B,l}) \gamma(i_{B,l+1})
 V_{C,i_{B,l}}^\dagger V_{C,i_{B,l+1}}$
is identity.
Then, we find that
\begin{align}
V_{C}:=V_{C,i_{B,0}}=\gamma(i_{B,l}) V_{C,i_{B,l}}.
\end{align}
That is, we have 
\begin{align}
V_{C}=\gamma(i_{B}) V_{C,i_{B}}.
\end{align}

{\bf Step 5:}\quad
For elements 
$i_B,i_B' \in {\cal I}_B$, 
the sets 
$\{ c_{i_{C}}\}_{i_{C} \in {\cal I}_{C,i_{B} }}$
and 
$\{c_{i_{C}'}\}_{i_{C}' \in {\cal I}_{C,i_{B}' }}$
span the space $\mathbb{C}^{d_C}$.
We choose 
$i_C \in {\cal I}_{C,i_B}$ and $i_C' \in {\cal I}_{C,i_B'}$ such that
$\langle c_{i_{C}},c_{i_{C}'}\rangle\neq 0$.
We have
\begin{align}
&\langle b_{i_{B}},b_{i_{B}'}\rangle \langle c_{i_{C}},c_{i_{C}'}\rangle
=\langle b_{i_{B}}\otimes c_{i_{C}},b_{i_{B}'}\otimes c_{i_{C}'}\rangle \\
=&
\beta(i_{B},i_{C})\beta(i_{B}',i_{C}')
\langle b_{i_{B}}'\otimes c_{i_{C}}',b_{i_{B}'}'\otimes c_{i_{C}'}'\rangle \\
=&
\langle b_{i_{B}}',b_{i_{B}'}'\rangle 
\beta(i_{B},i_{C})\beta(i_{B}',i_{C}')
\langle c_{i_{C}}',c_{i_{C}'}'\rangle \\
=&
\langle b_{i_{B}}',b_{i_{B}'}'\rangle 
\gamma(i_{B}) \gamma(i_{B}')
\langle V_{C} c_{i_{C}},V_{C} c_{i_{C}'}\rangle \\
=&
\gamma(i_{B}) \gamma(i_{B}')
\langle b_{i_{B}}',b_{i_{B}'}'\rangle 
\langle c_{i_{C}}, c_{i_{C}'}\rangle .
\end{align}
Since $\langle c_{i_{C}}, c_{i_{C}'}\rangle \neq 0$, we have
\begin{align}
\langle b_{i_{B}},b_{i_{B}'}\rangle 
=
\gamma(i_{B}) \gamma(i_{B}')
\langle b_{i_{B}}',b_{i_{B}'}'\rangle .
\end{align}
Also, we define the isometry $V_{B}: 
{\cal H}_B \to {\cal H}_B'$
such that
\begin{eqnarray}
V_{B} b_{i_B}= \gamma(i_{B})  b_{i_B}'
\end{eqnarray}
for $i_{B} \in {\cal I}_{B}$.

Therefore, for $(i_B,i_C) \in \cup_{i_B \in {\cal I}_B }(\{i_B\} \times {\cal I}_{C, i_B})$,
we have
\begin{align}
V_{BC}b_{i_B}\otimes c_{i_C}=b_{i_B}'\otimes c_{i_C}'=
(V_{B}\otimes V_C)b_{i_B}\otimes c_{i_C}.
\end{align}
Since $\{ b_{i_B}\otimes c_{i_C}\}_{(i_B,i_C) \in \cup_{i_B \in {\cal I}_B }(\{i_B\} \times {\cal I}_{C, i_B})}$
spans ${\cal H}_B \otimes {\cal H}_C$,
we have
\begin{align}
V_{BC}=V_{B}\otimes V_C.\Label{NVR2}
\end{align}
Combining \eqref {NVR} and \eqref{NVR2}, we have
\begin{align}
V=V_A\otimes V_{B}\otimes V_C.\Label{NVR3}
\end{align}
 \end{proof}

\subsubsection{General case}
We consider the general case.
We define 
$|v_i'\rangle:= \eta_i^{-1} \Pi_{i_A}^A\otimes \Pi_{i_B}^B
\otimes \Pi_{i_C}^C|\psi'\rangle$.

Let $\bar{\cal I}_A,\bar{\cal I}_B,\bar{\cal I}_C$ be the sets of indexes of the spaces 
${\cal H}_A,{\cal H}_B, {\cal H}_C$.

We introduce other conditions for 
the optimal maximizer given in \eqref{MA8} as a generalization of A3 and A4.

\begin{description}
\item[A8] Ideal systems ${\cal H}_A$, ${\cal H}_B$, and ${\cal H}_C$ are two-dimensional.
\item[A9] Each system has only two measurements.
That is, the sets $\bar{\cal I}_A$, $\bar{\cal I}_B$, and $\bar{\cal I}_C$
is composed of
4 elements.
For any element $i_A \in \bar{\cal I}_A$ ($i_B \in \bar{\cal I}_B$, $i_C \in \bar{\cal I}_C$), there exists 
an element $i_A' \in \bar{\cal I}_A$ ($i_B' \in \bar{\cal I}_B$, $i_C' \in \bar{\cal I}_C$) such that 
$\langle a_{i_A}| a_{i_A'}\rangle=0$
($\langle b_{i_B}| b_{i_B'}\rangle=0$, $\langle c_{i_C}| c_{i_C'}\rangle=0$).
\end{description}

Mermin Self-testing satisfies Conditions A8 and A9 in addition to Conditions A5, A6, and A7.

When A3 and A4 hold,
$\bar{\cal I}_A$ ($\bar{\cal I}_B$, $\bar{\cal I}_C$) is written as
${\cal B}_{A,0}\cup {\cal B}_{A,1}$ (${\cal B}_{B,0}\cup {\cal B}_{B,1}$, 
${\cal B}_{C,0}\cup {\cal B}_{C,1}$), where
${\cal B}_{A,j}=\{ (0,j),(1,j)\}$ (${\cal B}_{B,j}=\{ (0,j),(1,j)\}$, ${\cal B}_{C,j}=\{ (0,j),(1,j)\}$) 
and $\langle a_{(0,j)}| a_{(1,j)}\rangle=0$ 
($\langle b_{(0,j)}| b_{(1,j)}\rangle=0$, $\langle c_{(0,j)}| c_{(1,j)}\rangle=0$) 
for $j=0,1$.

We also consider the following condition for $\Pi_i= \Pi_{i_A}^A\otimes \Pi_{i_B}^B\otimes \Pi_{i_C}^C$.
\begin{description}
\item[C1]
When $i_A,i_A' \in \bar{\cal I}_A$ ($i_B,i_B' \in \bar{\cal I}_B$, 
$i_C,i_C' \in \bar{\cal I}_C$) satisfy
$\langle a_{i_A}| a_{i_A'}\rangle=0$ 
($\langle b_{i_B}| b_{i_B'}\rangle=0$, $\langle c_{i_C}| c_{i_C'}\rangle=0$), 
we have $\Pi_{i_A}^A+\Pi_{i_A'}^A=I$ ($\Pi_{i_B}^B+\Pi_{i_B'}^B=I$, 
$\Pi_{i_C}^C+\Pi_{i_C'}^C=I$). 
\end{description}

Let ${\cal H}_{i_A}^A$, ${\cal H}_{i_B}^B$, and ${\cal H}_{i_C}^C$ be the image of the projections
$\Pi_{i_A}^A$, $\Pi_{i_B}^B$, and $\Pi_{i_C}^C$.

\begin{theorem}\Label{APS2}
Assume that the optimal maximizer given in \eqref{MA8} satisfies
conditions A5, A6, A5, A7, A8, and A9, and 
the vectors $(\Pi_i |\psi'\rangle)_{i \in {\cal I}}$ realize the optimal solution in the SDP \eqref{thetaSDP}.
Then, there exist isometries 
$V_A$ from $ {\cal H}_A\otimes  {\cal K}_A$ to $ {\cal H}_A'$,
$V_B$ from $ {\cal H}_B\otimes  {\cal K}_B$ to $ {\cal H}_B'$,
and 
$V_C$ from $ {\cal H}_C\otimes  {\cal K}_C$ to $ {\cal H}_C'$
such that
\begin{align}
V_A \otimes V_B\otimes V_C |\psi\rangle \otimes |junk\rangle &=
|\psi'\rangle,\Label{ML5}\\
V_A \otimes V_B \otimes V_C|v_i\rangle \otimes |junk\rangle &=|v_i'\rangle , \Label{ML6}
\end{align}
for $i \in {\cal I}$,
where
$|junk\rangle$ is a state on ${\cal K}_A\otimes  {\cal K}_B\otimes  {\cal K}_C$.
\hfill $\square$\end{theorem}

\begin{proof}

Similar to the proof of Theorem \ref{APS}, we 
define orthogonal projections
$\bar\Pi_{j_X}^X $ on ${\cal H}_X$
such that the projection $\Pi^X:= \sum_{j_X}\bar\Pi_{j_A}^X$ satisfies 
$\Pi^X\psi'=\psi' $ for $X=A,B,C$.
Then, we define 
the projection $\bar\Pi_{(j_A,j_B,j_C)}:=\bar\Pi_{j_A}^A\bar\Pi_{j_B}^B\bar\Pi_{j_C}^C$.
In the same way as the proof of Theorem \ref{APS}, we define
$\alpha_{(j_A,j_B,j_C)}$,$\beta_{j_A,j_B,j_C}$, and 
$V_{X,j_X}$ for $X=A,B,C$.

We define the space ${\cal K}_X$ spanned by $\{|j_X\rangle\} $ for $X=A,B,C$.
We define the junk state on ${\cal K}_A\otimes {\cal K}_B\otimes {\cal K}_C$ as
\begin{align}
|junk\rangle:= \sum_{ j_A,j_B,j_C} \beta_{j_A,j_B,j_C}^{-1} \alpha_{(j_A,j_B,j_C)}^{-1}
|j_A,j_B,j_C\rangle.
\end{align}
We define the isometries $V_X: {\cal H}_X\otimes {\cal K}_X\to {\cal H}_X' $ 
as
\begin{align}
V_X:= \sum_{j_X} V_{X,j_X} \langle j_X|
\end{align}
for $X=A,B,C$.
The isometries $V_A$, $V_B$, and $V_C$ satisfy conditions \eqref{ML5} and \eqref{ML6}.
\end{proof}

\end{document}